\documentclass[12pt]{article}
\usepackage{amsmath}
\usepackage{graphicx}
\usepackage{enumerate}
\usepackage{natbib}
\usepackage{url} 

\usepackage{amsfonts,amsthm}
\usepackage{bbm}
\usepackage{float}
\usepackage{color}
\usepackage{verbatim}
\usepackage{numprint}

\usepackage{xr}
\newcommand{\blind}{1}

\newcommand{\indep}{\;\, \rule[0em]{.03em}{.67em} \hspace{-.25em}\rule[0em]{.65em}{.03em} \hspace{-.25em}\rule[0em]{.03em}{.67em}\;\,}
\newcommand{\Appendix}{\def\thesection{Appendix~\Alph{section}}\def\thesubsection{A.\arabic{subsection}}}

\newcommand{\ie}{\mathbbm{1}(E=e)}

\newcommand{\cp}{\stackrel{\mathcal{P}}{\rightarrow}}
\newcommand{\cl}{\stackrel{\mathcal{D}}{\rightarrow}}
\newcommand{\condind}{\perp\hspace{-1em}\perp}
\newcommand{\mystrut}{\vphantom{\int_0^1}}
\newcommand{\p}{\stackrel{p}{\rightarrow}}
\newcommand{\sumi}{\ensuremath{\sum_{i=1}^{n}}}
\newcommand{\sumj}{\ensuremath{\sum_{j=1}^{n}}}
\newcommand{\sumk}{\ensuremath{\sum_{k=1}^{n}}}
\newcommand{\ifB}{ \frac{\mathbbm{1}_e(\Tilde{e})}{\pi_e(\Tilde{x})} \{ \Tilde{y}- M_e( \Tilde{x}) \} + M_e( \Tilde{x}) - \mathbb{E}(M_e( X)) }

\newcommand{\ifAa}{ \mathbb{E}\left[ \mathbb{E}_{Y,I|X}\left[ \left. Y \cdot \left( \frac{\mathbbm{1}_e(\Tilde{e})}{\pi_e(\Tilde{x})} \{ \mathbbm{1}\{ \Tilde{j} \leq I \}- N_e(I, \Tilde{x}) \} + N_e(I, \Tilde{x}) - \mathbb{E}_X(N_e(I, X)) \right) \right| X, E=e \right] \right]    }

\newcommand{\ifAb}{ \frac{\mathbbm{1}_e(\Tilde{e})}{\pi_e(\Tilde{x})} \left\{ \Tilde{y} F(\Tilde{i})- \mathbb{E}(YF(I)|X= \Tilde{x},E=e) \right\} + \mathbb{E}(YF(I)|X= \Tilde{x},E=e) - \mathbb{E}(\mathbb{E}(YF(I)|X,E=e)) }

\newcommand{\Bhat}{\mathbb{E}(M_e( X)) }

\newcommand{\Ahat}{\mathbb{E}(\mathbb{E}(YF(I)|X,E=e)) }

\newcommand{\ifBi}{ \frac{\mathbbm{1}_e(\Tilde{e})}{\pi_e(\Tilde{x})} \{ \Tilde{y}- M_e( \Tilde{x}) \} + M_e( \Tilde{x}) - \mathbb{E}(M_e( X)) }

\newcommand{\ifAai}{ \mathbb{E}\left[ M_e(X) \cdot \mathbb{E}_{Y,I|X}\left[ \left. \left( \frac{\mathbbm{1}_e(\Tilde{e})}{\pi_e(\Tilde{x})} \{ \mathbbm{1}\{ \Tilde{j} \leq I \}- N_e(I, \Tilde{x}) \} + N_e(I, \Tilde{x})  \right) \right| X, E=e \right] \right]    }

\newcommand{\ifAci}{ - 2 \mathbb{E}(M_e(X)\mathbb{E}(F(I)|X,E=e)) }

\newcommand{\ifAbi}{ \frac{\mathbbm{1}_e(\Tilde{e})}{\pi_e(\Tilde{x})} \left\{ \Tilde{y} F(\Tilde{i})- M_e(\Tilde{x}) \mathbb{E}(F(I)|X= \Tilde{x},E=e) \right\} + M_e(\Tilde{x})\mathbb{E}(F(I)|X= \Tilde{x},E=e)  }

\newcommand{\Bhatn}{\mathbb{E}_n(M_e( X)) }

\newcommand{\Ahatn}{\mathbb{E}_n(\mathbb{E}_{\mathcal P_n}(YF(I)|X,E=e)) }

\newcommand{\ifBn}{ \frac{\mathbbm{1}_e(\Tilde{e})}{\pi_e(\Tilde{x})} \{ \Tilde{y}- M_e( \Tilde{x}) \} + M_e( \Tilde{x}) - \mathbb{E}(M_e( X)) }

\newcommand{\ifAan}{ \mathbb{E}_{\hat{\mathcal P}_n}\left[ \hat{M}_e(X) \cdot \mathbb{E}_{\hat{\mathcal P}_n}\left[ \left. \left( \frac{\mathbbm{1}_e(\Tilde{e})}{\pi_e(\Tilde{x})} \{ \mathbbm{1}\{ \Tilde{j} \leq I \}- \hat{N}_e(I, \Tilde{x}) \} + \hat{N}_e(I, \Tilde{x})  \right) \right| X, E=e \right] \right]    }

\newcommand{\ifAcn}{ - 2 \mathbb{E}_n(M_e(X)\mathbb{E}_{\mathcal P_n}(F(I)|X,E=e)) }

\newcommand{\ifAbn}{ \frac{\mathbbm{1}_e(\Tilde{e})}{\pi_e(\Tilde{x})} \left\{ \Tilde{y} \hat{F}(\Tilde{i})- \hat{M}_e(\Tilde{x}) \mathbb{E}_{\hat{\mathcal P}_n}(\hat{F}(I)|X= \Tilde{x},E=e) \right\} + \hat{M}_e(\Tilde{x})\mathbb{E}_{\hat{\mathcal P}_n}(\hat{F}(I)|X= \Tilde{x},E=e)  }

\renewcommand{\d}{\stackrel{d}{\rightarrow}}
\renewcommand{\thefootnote}{\fnsymbol{footnote}}

\newtheorem{Th}{\underline{\bf Theorem}}

\newtheorem{Rem}{\underline{\bf Remark}}

\newtheorem{Pro}{Proposition}

\newtheorem{Assump}{{Assumption}}

\def\boxit#1{\vbox{\hrule\hbox{\vrule\kern6pt  \vbox{\kern6pt#1\kern6pt}\kern6pt\vrule}\hrule}}
\def\bse{\begin{eqnarray*}}
\def\ese{\end{eqnarray*}}
\def\be{\begin{eqnarray}}
\def\ee{\end{eqnarray}}
\def\bsq{\begin{equation*}}
\def\esq{\end{equation*}}
\def\bq{\begin{equation}}
\def\eq{\end{equation}}
\def\th{^{th}}
\def\var{\hbox{var}}
\def\cov{\hbox{cov}}
\def\corr{\hbox{corr}}
\def\trace{\hbox{trace}}
\def\wh{\widehat}
\def\wt{\widetilde}
\def\eff {_{\rm eff}}
\def\sub{{\rm sub}}
\def\cat{{\rm cat}}
\def\th{^{\rm th}}
\def\my{\mathcal Y}
\def\eLL{\mathcal L}
\def\mR{\mathbb{R}}
\def\n{\nonumber}
\def\bias{\mbox{bias}}
\def\vecl{\mbox{vecl}}
\def\AIC{\mbox{AIC}}
\def\BIC{\mbox{BIC}}
\def\MSE{\mbox{MSE}}
\def\rank{\mbox{rank}}
\def\cov{\mbox{cov}}
\def\corr{\mbox{corr}}
\def\vec{\mbox{vec}}
\def\argmin{\mbox{argmin}}
\def\argmax{\mbox{argmax}}
\def\diag{\mbox{diag}}
\def\tr{\mbox{trace}}
\def\sumi{\sum_{i=1}^n}
\def\sumj{\sum_{j=1}^n}
\def\trans{^{\rm T}}
\def\hDash{\bot\!\!\!\bot}
\def\calS{\mbox{ $\mathcal{S}$}}
\def\calT{\mbox{ $\mathcal{T}$}}
\def\ba{{\boldsymbol\alpha}}
\def\ha{\widehat{\ba}}
\def\bb{{\boldsymbol\beta}}
\def\bd{{\boldsymbol\delta}}
\def\beps{\boldsymbol\epsilon}
\def\bg{{\boldsymbol\gamma}}
\def\bphi{{\boldsymbol\phi}}
\def\bzeta{{\boldsymbol\zeta}}
\def\bmu{\boldsymbol\mu}
\def\hb{\widehat{\bb}}
\def\he{\widehat{\varepsilon}}
\def\cs{\calS_{Y\mid\x}}
\def\cms{\calS_{\EE(Y\mid\x)}}
\def\defby{\stackrel{\mbox{\textrm{\tiny def}}}{=}}
\def\as{\mbox{ a.s.}}
\def\A{{\bf A}}
\def\a{{\bf a}}
\def\B{{\bf B}}
\def\c{{\bf c}}
\def\C{{\bf C}}
\def\D{{\bf D}}
\def\E{{\mathbb{E}}}
\def\V{{\bf V}}
\def\g{{\bf g}}
\def\f{{\bf f}}
\def\h{{\bf h}}
\def\b{{\bf b}}
\def\I{{\bf I}}
\def\M{{\bf M}}
\def\N{\mbox{ $\mathcal{N}$}}
\def\K{{\bf K}}
\def\t{{\bf t}}
\def\T{{\bf T}}
\def\bP{{\bf P}}
\def\bQ{{\bf Q}}
\def\bV{{\bf V}}
\def\hQ{{\widehat \bQ}}
\def\bS{{\bf S}}
\def\bs{{\bf s}}
\def\U{{\bf U}}
\def\u{{\bf u}}
\def\v{{\bf v}}
\def\W{{\bf W}}
\def\w{{\bf w}}
\def\X{{\bf X}}
\def\H{{\bf H}}
\def\q{{\bf q}}
\def\O{{\bf O}}
\def\x{{\bf x}}
\def\tx{{\widetilde \x}}
\def\Y{{\bf Y}}
\def\tY{{\widetilde Y}}
\def\y{{\bf y}}
\def\Z{{\bf Z}}
\def\z{{\bf z}}
\def\Ybar{{\overline{Y}}}
\def\xbar{{\overline{\x}}}
\def\bSig{{\bf \Sigma}}
\def\bLam{{\bf \Lambda}}
\def\blam{{\boldsymbol\lambda}}
\def\diag{\hbox{diag}}
\def\dfrac#1#2{{\displaystyle{#1\over#2}}}
\def\VS{{\vskip 3mm\noindent}}
\def\refhg{\hangindent=20pt\hangafter=1}
\def\refmark{\par\vskip 2mm\noindent\refhg}
\def\naive{\hbox{naive}}
\def\itemitem{\par\indent \hangindent2\parindent \textindent}
\def\dist{\hbox{dist}}
\def\trace{\hbox{trace}}
\def\refhg{\hangindent=20pt\hangafter=1}
\def\refmark{\par\vskip 2mm\noindent\refhg}
\def\Normal{\hbox{Normal}}
\def\Uniform{\hbox{Uniform}}
\def\povr{\buildrel p\over\longrightarrow}
\def\ccdot{{\bullet}}
\def\diag{\hbox{diag}}
\def\log{\hbox{log}}
\def\bias{\hbox{bias}}
\def\Siuu{\boldSigma_{i,uu}}
\def\squarebox#1{\hbox to #1{\hfill\vbox to #1{\vfill}}}
\def\btheta{{\boldsymbol \theta}}
\def\balpha{{\boldsymbol \alpha}}
\def\bpi{{\boldsymbol \pi}}
\def\bx{{\bf x}}
\def\0{{\bf 0}}
\def\1{{\bf 1}}
\def\vec{\mathrm{vec}}
\def\mA{\mathcal{A}}
\def\mB{\mathcal{B}}
\def\mC{\mathcal{C}}
\def\mH{\mathcal{H}}
\def\my{\mathcal Y}
\def\mp{\mathcal P}
\def\EE{\mathbb{E}}
\def\dfrac#1#2{{\displaystyle{#1\over#2}}}
\def\VS{{\vskip 3mm\noindent}}
\def\refhg{\hangindent=20pt\hangafter=1}
\def\refmark{\par\vskip 2mm\noindent\refhg}
\def\itemitem{\par\indent \hangindent2\parindent \textindent}
\def\var{\hbox{var}}
\def\cov{\hbox{cov}}
\def\corr{\hbox{corr}}
\def\trace{\hbox{trace}}
\def\refhg{\hangindent=20pt\hangafter=1}
\def\Normal{\hbox{Normal}}
\def\povr{\buildrel p\over\longrightarrow}
\def\ccdot{{\bullet}}
\def\pr{\hbox{pr}}
\def\wh{\widehat}
\def\wt{\widetilde}
\def\diag{\hbox{diag}}
\def\log{\hbox{log}}
\def\bias{\hbox{bias}}
\def\whT{\widehat{\Theta}}
\def\diag{\hbox{diag}}
\def\logit{{\mbox{logit}}}
\def\macomment#1{\vskip 2mm\boxit{\vskip 2mm{\color{red}\bf#1} {\color{blue}\bf -- MA\vskip 2mm}}\vskip 2mm}
\def\delunacomment#1{\vskip 2mm\boxit{\vskip 2mm{\color{blue}\bf#1} {\color{blue}\bf -- XdL\vskip 2mm}}\vskip 2mm}
\def\leecomment#1{\vskip 2mm\boxit{\vskip 2mm{\color{purple}\bf#1} {\color{blue}\bf -- Lee\vskip 2mm}}\vskip 2mm}

\def\ph#1#2#3{\varphi{_#1}(#2,#3)}
\def\ps#1#2{\Psi_{#1}(#2)}
\def\pn{\hat{\mathcal P}_n}
\def\e#1#2{\mathbb{E}_{#1} \left[ #2 \right]}

\begin{document}

\def\spacingset#1{\renewcommand{\baselinestretch}%
{#1}\small\normalsize} \spacingset{1}

\date{}
\if1\blind
{
  \title{\bf Causal inference targeting a concentration index for studies of health inequalities}
  \author{Mohammad Ghasempour\thanks{
    We are grateful for comments given by Tetiana Gorbach which have improved the paper. The Marianne and Marcus Wallenberg Foundation is acknowledged for their financial support.}\hspace{.2cm}\\
    Department of Statistics/USBE, Ume\r{a} University, Sweden \\
    and \\
    Xavier de Luna \\
    Department of Statistics/USBE, Ume\r{a} University, Sweden 
    \\
    and \\
    Per E. Gustafsson \\
    Department of Epidemiology and Global Health, Ume\r{a} University, Sweden}
  \maketitle
} \fi

\if0\blind
{
  \bigskip
  \bigskip
  \bigskip
  \begin{center}
    {\LARGE\bf Causal inference targeting a concentration index for studies of health inequalities}
\end{center}
  \medskip
} \fi

\bigskip
\begin{abstract}
A concentration index, a standardized covariance between a health variable and relative income ranks, is often used to quantify income-related health inequalities. There is a lack of formal approach to study the effect of an exposure, e.g., education, on such measures of inequality. In this paper we contribute by filling this gap and developing the necessary theory and method. Thus, we define a counterfactual concentration index for different levels of an exposure. We give conditions for the identification of this complex estimand, and then deduce its efficient influence function. This allows us to propose estimators, which are regular asymptotic linear under certain conditions. In particular, we show that these estimators are $\sqrt n$-consistent and asymptotically normal, as well as locally efficient. The implementation of the estimators is based on the fit of several nuisance functions. The estimators proposed have rate robustness properties allowing for convergence rates slower than $\sqrt{n}$-rate for some of the nuisance function fits.
The relevance of the asymptotic results for finite samples is studied with simulation experiments. We also present a case study of the effect of education on income-related health inequalities for a Swedish cohort.
\end{abstract}

\noindent%
\vfill

\newpage
\spacingset{1.9} 
\npdecimalsign{.} 
\nprounddigits{3}

\section{Introduction}
In the health inequality literature, a concentration index is often used to quantify \linebreak socioeconomic-related inequality in a given health variable $Y$ \citep[see, e.g.,][and citations therein]{ataguba:2022,worldbank:2008}. If a continuous variable $I$ (we use income in this paper), is used as a socioeconomic measure, the concentration index 
\begin{align}\label{naive.eq}
    G=\frac{2cov(Y,F_I(I))}{\EE(Y)},
\end{align}
is a standardized (valued between -1 and 1) covariance between the health variable $Y\geq 0$ and the relative income ranks $F_I(I)$; see, e.g., \citet{koolman:2004}. In a case study described in detail in Section \ref{app.sec}, we use register data on all those born in Sweden during year 1950, and obtain an estimate (using observed relative ranks and sample moments) $\widehat G=\numprint{-0.37934}$ (s.e. $\numprint{0.007235}$) when $Y$ is an indicator of death between age 50 and 60, and $I$ is the average yearly income while alive during the same time period. The observed negative sign for $\widehat G$ can be interpreted as income-related inequality in survival for the 1950 Swedish cohort.  
This kind of inequality has been observed for many health indicators earlier in the literature (\citet{Wallar2020,Sommer2015,Schwendicke2014-ir,Ngamaba2017-lq}). While quantifying inequalities is of interest per say, an important question is to understand their causes. A hypothesis in this context is that increasing educational level in a population could potentially be used to decrease income-related inequalities in health (\citet{Gerdtham2016,HECKLEY201689}). 
Thus, if $E$ denotes education level, a counterfactual scientific question of interest is whether exposing a group of individuals to a given higher level of education would lower inequality in health (less negative concentration index). A naive investigation of this hypothesis is to estimate $G$ separately for individuals having different education levels. However, because of confounding this would most probably give a biased answer to the above counterfactual question of interest.

A related literature studies causal pathways explaining health disparities, where for instance, one may be interested in how education can affect health through income where the latter is then a mediator; see \citet{jackson:2021,jacksonetal:2018}, and references therein. This is, however, a different question to the one addressed here where a concentration index considers health and income observed simultaneously in time, and thus income is not on the causal pathway (a mediator) between education and health. 

In this paper, we thus contribute by developing the theory and methods in order to answer counterfactual questions on concentration index \eqref{naive.eq}. Note that \citet[][Sec. 20.3.2]{ImbensRubin:book} discussed a counterfactual Gini index \citep{milanovic:1997}, which is a specific case of \eqref{naive.eq}, where $Y \equiv I$. As Imbens and Rubin, we use here the potential outcome framework \citep{rubin74,holland86} to define a target counterfactual concentration index corresponding to \eqref{naive.eq}. We contribute by giving conditions for identification, and deducing the efficient influence function (EIF) of this complex target estimand. This allows us to propose estimators of the counterfactual concentration index, which are regular asymptotic linear under certain conditions \citep{Vaart_1998,Demystifying,kennedy2016semiparametric}. In particular, these estimators are $\sqrt n$-consistent and asymptotically normal, as well as locally efficient. The implementation of the estimators is based on the fit of several nuisance functions: e.g., regressions of the health and exposure variables on covariates, as well as counterfactual ranks of the income variable. The estimators proposed have rate robustness properties allowing for convergence rates slower than $\sqrt{n}$-rate for some of the nuisance function fits. These results are, however, more complex than similar known rate robustness results for simpler causal estimands such as the widely studied average causal effect of an exposure on an outcome of interest \citep[see, e.g., the reviews by][]{smucler2019unifying,kennedy2016semiparametric,moosavi2021}. We operationalize two different EIF-based estimators, which necessitate atypical procedures to fit some of the nuisance functions involved. The proposed estimators are compared with the superefficient plug-in estimator in simulation experiments, where robustness properties to misspecification of nuisance models are also investigated. 

The paper is organized as follows. Section \ref{thm.sec} introduces the target estimand, conditions for identification, deduces the efficient influence functions, propose EIF-based estimators and give asymptotic properties. Section \ref{sim.sec} uses simulation experiments to study finite sample properties of the estimators proposed. Section \ref{app.sec} presents a case study where the effect of education on income-related health inequality for Swedes born in 1950 is studied. Section \ref{conc.sec} concludes the paper.

\section{Theory and method}\label{thm.sec}
\subsection{Assumptions and identification}
For a sample of $n$ individuals, we observe two outcomes of interest: a health variable $Y\geq 0$ (e.g., survival, disease incidence, hospitalization, etc.), and a socioeconomic variable $I$ (continuous, typically income), which we want to relate to study health inequality. Further, we consider an exposure of interest (categorical, e.g., education level $E$), as well as a vector of observed pre-exposure variables $X$. Let $E=0,1,\ldots,K-1$, i.e. $K$ levels of exposure are observed. The observed data distribution of $O=(Y,I,E,X)$ is denoted $\mp$.
For the outcome variables, we define their potential outcomes \citep{rubin74} under each exposure level: $Y(e)$ and $I(e)$, $e=0,\ldots,K-1$. For each individual in the sample, only one pair $(Y(e),I(e))$ of these potential outcomes will be observed. We assume consistency, that is the observed outcomes are $Y_i=Y_i(e)$ and $I_i=I_i(e)$ if $E_i=e$, for $i=1,\ldots,n$, and that there is no interference between individuals \citep[see,][]{rubin1991}.

We use a concentration index valued between -1 and 1 to measure health inequality with respect to a socioeconomic variable:
\begin{align*}
	G&=\frac{2cov(Y_i,F_I(I_i))}{\EE(Y_i)} =2\frac{\EE(Y_iF_I(I_i))}{\EE(Y_i)}-2\EE(F_I(I_i))=2\frac{\EE(Y_iF_I(I_i))}{\EE(Y_i)}-1,\\
	\end{align*}
 where $\EE(Y_i)> 0$, and $F_I(\cdot)$ is the cumulative distribution function of $I_i$ \linebreak  \citep[e.g.,][]{koolman:2004}. 
 Interpretation of the index $G$ corresponds to the Gini index \citep{milanovic:1997} when $Y\equiv I$. 

We are interested in the causal effects of an exposure $E$ 
and we use the potential outcomes to define counterfactual versions of the concentration index as follows:
\begin{align}
	G(e)=2\frac{\EE(Y_i(e)F_e(I_i(e)))}{\EE(Y_i(e))}-1, 
\end{align}
where expectations are w.r.t. the distributions of the random variables $(Y(e), I(e))$, with $\EE(Y_i(e))>0$, and $F_e(\cdot)$ is the cumulative distribution function of $I_i(e)$.
The causal parameters of interest are the contrasts $$\theta(e)=G(e)-G(0),\hskip 1cm e=1,\ldots,K-1.$$ These correspond to the causal effects of different exposure levels.
In an observational study, where the exposure is not randomized, we make the following assumptions to obtain identification.
\begin{Assump}\label{unconfoundedness.ass} For $e=0,\ldots,K-1$,
    \begin{itemize}
        \item[i)]  $Y(e) \indep E \mid X$,
        \item[ii)] $I(e) \indep E \mid X$,  
    \end{itemize}
    where "$\indep$" stands for "independent" \citep{dawid1979conditional}. 
\end{Assump}
\begin{Assump}\label{alternative1.ass} 
 $Y(e)\indep I(e)\mid X$, for $e=0,\ldots,K-1$.    
\end{Assump}
\begin{Assump}\label{alternative2.ass} 
    $ Y(e) F_{e}(I(e))  \indep E \mid X $, for $e=0,\ldots,K-1$.
\end{Assump}
\begin{Assump}
\label{overlap.ass}
	$\pi_e(x):=\Pr(E=e\mid X)>0$, for $e=1,\ldots,K-1$.
\end{Assump}


	Assumption \ref{unconfoundedness.ass} states that given a set of observed covariates, the exposure is ignorable for the health outcome $Y(e)$ and the socioeconomic outcome $I(e)$.
		Assumption \ref{overlap.ass} guarantees that all individuals in the study can be exposed to the different levels of $E$, and are therefore comparable. These are typical assumptions in observational studies, unless other identification information is available, e.g. the existence of mediators or instruments \citep[e.g.,][]{Gorbach:2023,fulcher2020robust}. 
	Finally, Assumptions \ref{alternative1.ass} and \ref{alternative2.ass} are two different alternative assumptions for the identification of the numerator of $G(e)$. 
		
Note, that a sufficient condition for Assumption \ref{alternative2.ass} to hold is $\{Y(e),I(e)\} \indep E \mid X$, for $e=0,\ldots,K-1$ (a slightly stronger version of Assumption \ref{unconfoundedness.ass}). Assumption \ref{alternative1.ass} may be less realistic than Assumption \ref{alternative2.ass}. However, we consider it as an alternative because it yields simpler estimators, which happen to be robust and behave well when Assumption \ref{alternative2.ass} holds instead of \ref{alternative1.ass}; see Remark \ref{miss.rem}. 

	\begin{Pro}[Identification]\label{identification.prop}
		Let Assumptions \ref{unconfoundedness.ass}  and \ref{overlap.ass} hold, as well as Assumption \ref{alternative1.ass} or \ref{alternative2.ass}. Then, $G(e)$, for $e=1,\ldots,K-1$, and thereby $\theta(e)$, are identified, where the following functionals of the observed data distribution $\mp$ identify $G(e)$, respectively for $j=1$ (Assumption \ref{alternative1.ass}) and $j=2$ (Assumption \ref{alternative2.ass}):
\begin{align}
    G(e)\equiv\Psi_j(\mp)&=2\frac{A_j(\mp)}{B(\mp)}-1, \label{Psi.eq}
\end{align}
with
\begin{align}
    A_1(\mp)&=\EE\{\EE(Y\mid X, E=e) \EE (\Xi_{e}(I,\mp)  \mid X, E=e ) \}, \label{A1.eq} \\
    A_2(\mp)&=\EE \{\EE (Y\Xi_{e}(I,\mp) \mid X, E=e)\}, \label{A2.eq}
\end{align}
and $B(\mp)= \EE \{\EE (Y\mid X, E=e) \}$, 
where $\Xi_{e}(i,\mp)=\EE \{\EE (\mathbbm{1}(I\leq i) \mid X, E=e ) \}$.
	\end{Pro}
		\begin{proof}
  In the sequel we will not index the expectation operator $\EE$ as long as the context is clear.
		That the identification of $\EE(Y_i(e))$ with $B(\mp)$ follows from Assumptions \ref{unconfoundedness.ass} and \ref{overlap.ass} is a classical result in the causal inference literature \citep[e.g.,][Chap 7]{HernanRobins:book}. We focus here on the numerator $\EE(Y_i(e)F_e(I_i(e)))$ and show that it is identified by $A_j$, for $j=1,2$. We show first that $F_e(i)$ is identified for a given value $i$. 
	 
		We can write:
			\begin{align*}
			   F_e(i)&= \EE (\mathbbm{1}(I(e)\leq i)),
			\end{align*}
   where $\mathbbm{1}(\cdot)$ is the indicator function. 
			Further,
			\begin{align*}
            F_e(i)&=\EE\{\EE(\mathbbm{1}(I(e)\leq i) \mid X )\}, \\
			   &=\EE\{\EE(\mathbbm{1}(I(e)\leq i) \mid X, E=e) \}, \\
			   &=\EE\{\EE(\mathbbm{1}(I\leq i) \mid X, E=e ) \}=:\Xi_{e}(i,\mp).
			\end{align*}
   The first equality follows from the law of total expectation, and the second equality is due to Assumptions \ref{unconfoundedness.ass}ii) and \ref{overlap.ass}, and the last one is due to consistency. The last expression is a functional of the observed data distribution $\mp$. 

		Consider now first the case when Assumption \ref{alternative1.ass} holds, then
			\begin{align*}
				\EE(Y(e)F_e(I(e)))&= \EE\{\EE(Y(e) F_e(I(e)) \mid X)\}\\
				&=\EE\{\EE(Y(e) \mid X) \EE( F_e(I(e)) \mid X) \}.
			\end{align*}
			Using further consistency, no interference, and Assumptions \ref{unconfoundedness.ass} and \ref{overlap.ass} we can write using the above identification expression for $F(I(e))$: 
			\begin{align*}
			\EE(Y(e)F_e(I(e)))&=\EE\{\EE(Y\mid X, E=e) \EE( \Xi_{e}(I,\mp) \mid X, E=e) \}. 
			\end{align*}
			The latter expression is a functional of the observed data distribution $\mp$.
			
	  Consider, finally, the case when Assumption \ref{alternative2.ass} holds instead of \ref{alternative1.ass}. Then, we can write 
	  	\begin{align*}
				\EE(Y(e)F_e(I(e)))&= \EE\{\EE(Y(e) F_e(I(e)) \mid X, E=e)\}\\
				&=\EE\{\EE(Y\Xi_{e}(I,\mp) \mid X, E=e)\},
			\end{align*}
which is a functional of the observed data distribution.
		\end{proof}
Thus, by this proposition, the two statistical parameters $\Psi_1$ and $\Psi_2$, have a causal interpretation under Assumption \ref{alternative1.ass} and \ref{alternative2.ass} respectively. 
	
\subsection{Efficient influence function and estimation}
We use semiparametric theory \citep{Vaart_1998} to obtain locally efficient estimators of $\Psi_1$ and $\Psi_2$ with some robustness properties and in particular which allow for the use of machine learning algorithms for the estimation of the nuisance functions. 
For this purpose we deduce first the efficient influence function of the statistical parameters of interest.

\begin{Th}[Efficient influence functions]\label{eif:thm}
Let Assumptions \ref{unconfoundedness.ass} and \ref{overlap.ass} hold. Then, the efficient influence functions for $\Psi_1$ (under Assumption \ref{alternative1.ass}) and $\Psi_2$ (under Assumption \ref{alternative2.ass}) are, respectively, for observation $\tilde o=(\tilde y,\tilde i,\tilde e,\tilde x)$,
\begin{align*}
     \varphi_j(\tilde o,\mp)&= 2\frac{\ph{{A_j}}{\tilde{o}}{\mp}}{B(\mp)}-2\frac{\ph{B}{\tilde{o}}{\mp} A_j(\mp)}{B(\mp)^2},
\end{align*}
   $j=1,2$, where
  \begin{align*}
  \ph{{A_1}}{\tilde{o}}{\mp} = & \frac{\mathbbm{1}(\tilde e=e)}{\pi_e( \tilde{x})} [\tilde{y} \Xi_{e}(\tilde{i},\mp)  -  \EE(Y \mid  \tilde{x},E=e)\EE( \Xi_{e}(I,\mp) \mid  \tilde{x},E=e)]  \\
  &+\EE(Y \mid  \tilde{x},E=e)\EE( \Xi_{e}(I,\mp) \mid  \tilde{x},E=e) -  A_1(\mp)\\
  &+ \EE [ \EE \{  Y\mid   X,E=e \} \EE\{\varphi_\Xi (I,\tilde o,\mp) \mid   X,E=e \} ], \\
\ph{{A_2}}{\tilde{o}}{\mp} = &
\frac{\mathbbm{1}(\tilde e=e)}{\pi_e( \tilde{x})} [\tilde{y} \Xi_{e}(\tilde{i},\mp)  -  \EE(Y  \Xi_{e}(I,\mp) \mid  \tilde{x},E=e)]  +\EE(Y \ \Xi_{e}(I,\mp) \mid  \tilde{x},E=e) -  A_2(\mp) \\
 &+ \EE [ \EE \{  Y \varphi_\Xi (I,\tilde o,\mp) \mid   X,E=e \} ], \\
    \varphi_B(\tilde{o},{\mp})=&\frac{\mathbbm{1}(\tilde e = e) }{\pi_e( \tilde{x})}\left(\tilde{y}-\EE\left(Y \mid \tilde{x}, E=e \right)\right)+ \EE\left(Y |  \tilde{x},E=e\right) - B(\mp),
\end{align*}  
with 
\begin{align*}
    \varphi_\Xi (j,\tilde o,\mp)=\frac{\mathbbm{1}(\tilde e=e)}{\pi_e(\tilde{x})} [ \mathbbm{1}(\tilde{i} \leq j) - \EE(\mathbbm{1}(I \leq j) \mid  \tilde{x}, E=e) ] + \EE(\mathbbm{1}(I \leq j) \mid \tilde{x},E=e) - \Xi_{e}(j,\mp).
\end{align*}

\end{Th}
See Supplementary Material Section \ref{proofthm1.app} for a proof.

Using these influence functions we can construct well-behaved estimators. However, consider first the naive plug-in estimators $\Psi_1(\pn)$ and $\Psi_2(\pn)$, where we replace the observed data distribution with its empirical version $\pn$. This is done by using data driven fits of the nuisance regression functions in (\ref{Psi.eq}); see Supplementary Material Section \ref{estimation.app} for details.
The plug-in estimator is well-behaved when correctly specified parametric models for the nuisance functions are fitted (ensuring parametric $\sqrt n$-rate of convergence). On the other hand, if these parametric models are misspecified or flexible machine learning methods are used (slower rate of convergence), then the plug-in estimator is typically biased, see \citet{moosavi2021,kennedy2016semiparametric} for reviews. 

There are several ways of constructing well-behaved, i.e. regular and asymptotic linear, estimators using the efficient influence functions of Theorem \ref{eif:thm}, including one-step estimators, estimating equation estimators and targeted learning \citep[see, e.g.,][]{Demystifying,van2014targeted,farrell2015robust,chernozhukov2018double}. 
We focus here on the one-step estimator 
\begin{align}\label{1s.eq}
    \Psi_j(\pn)+\frac{1}{n}\sum_{i=1}^n \varphi_j(O_i,\pn),
\end{align}
and the
estimating equation estimator, which is the solution of
\begin{align}\label{eq.eq}
  \frac{1}{n}\sum_{i=1}^n \varphi_j(O_i,\pn)=0,  
\end{align}
see Supplementary Material Section \ref{estimation.app} for the operationalization and implementation details of these estimators. 
While estimators \eqref{1s.eq} and \eqref{eq.eq} are identical when the target estimand is the average causal effect of an exposure \citep{Demystifying}, this is not the case here.

\subsection{Asymptotic properties}

The estimators \eqref{1s.eq} and \eqref{eq.eq}
are bias corrections of $\Psi_j(\pn)$ justified by the von Mises expansion:
\begin{align} \label{vonM.eq}
    \sqrt{n}\left\{\left[\Psi_j(\pn)-\Psi_j(\mp)\right]+\EE_n\left[\varphi_j(O,\pn)\right]\right\} = &\sqrt{n}\EE_n\left[\varphi_j(O,\mp)\right] \\ \nonumber
    &+\sqrt{n}(\EE_n-\EE_\mp)\left[\varphi_j(O,\mp)-\varphi_j(O,\pn) \right] \\ \nonumber
    &-\sqrt{n}R_j(\mp,\pn),
\end{align}
where $\EE_n(O) := \frac{1}{n} \sum_{i=1}^n O_i.$ The central limit theorem can be applied to the first term on the right hand side under usual regularity conditions, yielding the asymptotic normal distribution 
\begin{align}\label{gauss.eq}
    N(0,\EE_\mp[\varphi(O,\mp)^2]),
\end{align}
while the last two terms need to be controlled for, i.e. we need conditions so that they are of order $o_\mp(1)$. When this is the case, the estimators are asymptotically normal and have variance that reaches the semiparametric efficiency bound $\EE_\mp[\varphi(O,\mp)^2]$ \citep[e.g.,][]{kennedy2016semiparametric}.

The empirical process term $\sqrt{n}(\EE_n-\EE_\mp)\left[\varphi_j(O,\mp)-\varphi_j(O,\pn) \right]$ is $o_\mp(1)$ if $\varphi_j(O,\pn)$ is weakly consistent for $\varphi_j(O,\mp)$ and $\varphi_j(O,\pn)$ belongs to a $\cal P$-Donsker class \citep[][Lemma 19.24]{Vaart_1998}. This Donsker condition is not needed if sample splitting is used to estimate the nuisance models in a separate split of the data \citep{zheng2011cross,chernozhukov2018double}.

Importantly, we need to show that the reminder term $\sqrt{n}R_j(\mp,\pn)$ is controlled. We use the following conditions.
\begin{Assump}\label{product.ass}
    \begin{itemize}
        \item[] 
        \item[i)]For the plug-in estimators with $j=1,2$: 
        $$ \left[\frac{B(\mp)}{B(\pn)} - 1 \right] \cdot [\ps{j}{\mp} -\ps{j}{\pn}] \asymp o_{\mp}(\sqrt{n}^{-1}).$$
        \item[ii)]  The plug-in estimator $B(\pn)^{-1}$ is bounded in probability:
        $$ \forall \varepsilon \quad \exists \delta_\varepsilon, N_\varepsilon \quad \mbox{such that}: P(|B(\pn)^{-1}| \geq \delta_\varepsilon) \leq \varepsilon \quad \forall n>N_\varepsilon .$$
        \item[iii)] For the nuisance functions $M_e({x})=\EE_\mp (Y\mid X=x, E=e)$ and $\pi_e(x)$:
        $$ \e{\mp}{ \left( \frac{\pi_e(X)}{\hat{\pi}_e(X)}  -1\right) (M_e(X) - \hat{M}_e(X) )  } \asymp o_{\mp}(\sqrt{n}^{-1}).$$
        \item[iv)] For the nuisance functions $N_e(j,x) = \mathbb{E}_\mp[\mathbbm{1} \{ I \leq j \} |X=x, E=e]$, $M_e({x})$ and $\pi_e(x)$:
        $$ \e{\mp}{ \left( \frac{\pi_e(X)}{\hat{\pi}_e(X)}  -1\right) \hat\EE_{\pn(\tilde X)}  \left[ \hat{M}_e(\Tilde{X})  \hat{\mathbb{E}}_{\pn(\tilde I \mid \tilde X)}\left[ \left. N_e(\tilde I, X)- \hat{N}_e(\tilde I, X) \right| \Tilde{X}, E=e \right]  \right]  } \asymp o_{\mp}(\sqrt{n}^{-1}).$$
        We use here and in the sequel $\hat\EE_{\pn}$ to denote the estimation of $\EE_\mp$.
        \item[v)] For the nuisance functions $M_e({x})$, $\pi_e({x})$ and $\Xi_{e}(i,\mp)$: 
        \begin{align*}
            \EE_{\mp} &\left[\left( \frac{\pi_e(X)}{\hat{\pi}_e(X)} -1 \right) \right. \\
            &\left.\cdot\left( M_e(X)\mathbb{E}_{\mathcal P}(\Xi_{e}(I,\pn)| X,E=e)- \hat{M}_e(X) \hat{\mathbb{E}}_{\hat{\mathcal P}_n (\tilde I \mid X) }(\Xi_{e}(\tilde I,\pn)|X,E=e) \right)\right]      \asymp o_{\mp}(\sqrt{n}^{-1}).
        \end{align*}
         \item[vi)] For the nuisance functions $N_e(j,x) = \mathbb{E}_\mp[\mathbbm{1} \{ I \leq j \} |X=x, E=e]$, $M_e({x})$ and $\pi_e(x)$:
        $$ \e{\mp}{ \left( \frac{\pi_e(X)}{\hat{\pi}_e(X)}  -1\right) \hat\EE_{\pn(\tilde X)}{\left[   \hat{\mathbb{E}}_{\pn(\tilde Y ,\tilde I \mid \tilde X)}\left[ \tilde Y (\left. N_e(\tilde I, X)- \hat{N}_e(\tilde  I, X) )\right| \tilde X, E=e \right]  \right]}  } \asymp o_{\mp}(\sqrt{n}^{-1}).$$
        \item[vii)] For the nuisance functions $M_e({x})$, $\pi_e({x})$ and $\Xi_{e}(i,\mp)$: 
        \begin{align*}
            \EE_{\mp} &\left[ \left( \frac{\pi_e(X)}{\hat{\pi}_e(X)} -1 \right) \right. \\
            &\left. \cdot \left( \mathbb{E}_{\mp}(Y \Xi_e(I,\pn)|X,E=e)-  \hat{ \mathbb{E}}_{\hat{\mathcal P}_n (\tilde Y, \tilde I \mid X) }(\tilde Y \Xi_e( \tilde I,\pn)|X,E=e) \right) \right]     \asymp o_{\mp}(\sqrt{n}^{-1}).
        \end{align*}
        \item[viii)]  $  Y(e)\cdot g(I(e))  \indep E \mid X $ for any cadlag function $g : R \to [0,1]$ satisfying $\lim_{x \to -\infty} g(x) =0$ and $\lim_{x \to \infty} g(x) =1$.
        \item[ix)] For the nuisance functional $\Lambda(h )$ which is defined by:
        $$\Lambda(h ) : = \e{\mp}{ M_e(X)  \e{\mp}{ h(I) | X , E=e }  -  \hat{M}_e( X) \hat\EE_{\pn( \tilde I \mid  X)}{\left[ h(\tilde I)  | X, E=e \right] }   } ,$$
        equipped with the norm $||\Lambda||_{sup} = \sup_{||h||\leq 1}{ |\Lambda (h)| } $, assume:
        $$ \EE_\mp{[(\Xi_e(I,\pn) - \Xi_e(I,\mp))^2]^\frac{1}{2}} ||\Lambda||_{sup}^\frac{1}{2} \asymp o_{\mp}(\sqrt{n}^{-1}).$$
        \item[x)] The quantities $\hat{M}_e(\tilde X) \hat\EE_{\pn( \tilde I \mid \tilde X)}{\left[ \Xi_e(\tilde I,\pn)  |\tilde X, E=e \right]}$ and $\hat{M}_e(\tilde X) \hat\EE_{\pn( \tilde I \mid \tilde X)}{\left[ \Xi_e(\tilde I,\mp)  |\tilde X, E=e \right]}$ belong to a Donsker class, and both are weakly consistent estimators of 
        
        $M_e(X) \EE_{\mp}{\left[ \Xi_e( I,\mp)  | X, E=e \right]}$.

        \item[xi)] For the nuisance functional $\Lambda_2(h )$, which is defined by
        $$\Lambda_2(h ) : = \e{\mp}{\e{\mp}{ Y h(I) | X , E=e }  -  \hat\EE_{\pn( \tilde Y, \tilde I \mid  X)}{\left[\tilde Y h(\tilde I)  | X, E=e \right] }   } ,$$
        equipped with the norm $||\Lambda_2||_{sup} = \sup_{||F||\leq 1}{ |\Lambda_2 (F)| }  $ , assume:
        $$ \EE_\mp{[(\Xi_e(I,\pn) - \Xi_e(I,\mp))^2]^\frac{1}{2}} ||\Lambda_2||_{sup}^\frac{1}{2} \asymp o_{\mp}(\sqrt{n}^{-1}).$$
        \item[xii)] The quantities $\hat\EE_{\pn( \tilde Y,\tilde I \mid \tilde X)}{\left[\tilde Y \Xi_e(\tilde I,\pn)  |\tilde X, E=e \right]}$ and $ \hat\EE_{\pn(\tilde Y, \tilde I \mid \tilde X)}{\left[\tilde Y \Xi_e(\tilde I,\mp)  |\tilde X, E=e \right]}$ belong to a Donsker class, and both are weakly consistent estimators of 
        
        $\EE_{\mp}{\left[Y \Xi_e( I,\mp)  | X, E=e \right]}$.
    \end{itemize}
\end{Assump}
Assumption \ref{product.ass}i) is a product rate condition for the plug-in estimators $B(\pn)$ and $\Psi_1(\pn)$. It allows both estimators to be estimated at a lower rate than $\sqrt{n}$.
Assumption \ref{product.ass}iii) is the usual product rate condition for the estimation of $B(\mp)$ \citep{farrell2015robust}.
Assumptions \ref{product.ass}iv)-vii) are further product rate conditions needed for the estimation of $A_j(\mp)$, $j=1,2$. Further, Assumption \ref{product.ass}viii) is a stronger version of the identification Assumption \ref{alternative2.ass}.
Finally, further rate conditions and Donsker conditions are used to control the remainder terms: Assumption \ref{product.ass}ix)-x) for $A_1(\mp)$ and Assumption \ref{product.ass}xi)-xii) for $A_2(\mp)$. 
As earlier, the Donsker conditions can be discarded if sample splitting is used to estimate the nuisance models in a separate split of the data.


Note that the remainder term can be written as:
\begin{equation*}
    \sqrt{n} R_j(\mp , \hat{\mp}_n) =\sqrt{n} \left( \Psi_j(\pn ) + \E_\mp \left[ \varphi_j (O, \hat{\mp}_n) \right] - \Psi_j( \mp ) \right),
\end{equation*}
and can be controlled as follows.
\begin{Th}[Controlling the remainder term]\label{remainder.thm}
    Under Assumptions \ref{unconfoundedness.ass}, \ref{alternative1.ass}, \ref{overlap.ass}, \ref{product.ass}i)-v), and \ref{product.ass}viii)-x) 
    $$\sqrt{n} \left( \Psi_1(\pn ) + \E_\mp \left[ \varphi_1 (O, \hat{\mp}_n) \right] - \Psi_1( \mp ) \right) \asymp o_{\mp}(1).$$ 
Further, under Assumptions  \ref{unconfoundedness.ass}, \ref{overlap.ass}, \ref{product.ass}i)-iii), \ref{product.ass}vi)-viii), and \ref{product.ass}xi)-xii) 
$$\sqrt{n} \left( \Psi_2(\pn ) + \E_\mp \left[ \varphi_2 (O, \hat{\mp}_n) \right] - \Psi_2( \mp ) \right) \asymp o_{\mp}(1).$$ 
\end{Th}
See Supplementary Material Section \ref{proofthm2.app} for a proof.

The direct consequence of this last result is that under the assumptions mentioned in this subsection the one-step estimator is regular asymptotic linear, and has asymptotically normal distribution \eqref{gauss.eq}. Note that this theorem also applies to the estimating equation estimator \eqref{eq.eq}, but where the assumptions have to hold for $\hat{\mp}^*_n$, a different empirical distribution. Because the latter is not explicit the assumptions are difficult to verify.

\begin{Rem}\label{miss.rem}
    The product conditions used in the theorem above imply the following rate robustness properties.
    If the nuisance functions $\pi_e(X)$ and $N_e(I,\cal P)$ are correctly specified (i.e. fitted at $\sqrt{n}$-rate) then the asymptotic results hold even if the outcome model $M_e(X)$ and other nuisance functions involving the outcome $Y$ are fitted at a slower rate. In such situations, the estimator based on Assumption \ref{alternative1.ass} will be robust, i.e. behave well even when only the weaker and more realistic Assumption \ref{alternative2.ass} holds.
    Furthermore, if we have a correct parametric model for the joint distribution of $(Y(e),I(e))\mid X$, then we can estimate the propensity score $\pi_e(X)$ at a slow rate. Finally, the product conditions in Assumptions \ref{product.ass}ix) and xi) imply that root-$n$ convergence for one of the two terms in the products is needed.  
\end{Rem}

\section{Simulation study}\label{sim.sec}
We perform numerical experiments in order to study the finite sample properties of the estimators proposed. 
We consider an exposure $E$ with three levels, for which we define three potential incomes $I(0)$, $I(1)$, and $I(2)$, as well as three potential health-related outcomes $Y(0)$,$ Y(1)$, and $Y(2)$. All potential outcomes are assumed to be determined by values of two covariates $X_1 \sim N(1,1)$, and $X_2 \sim N(10,1)$ plus additive normally distributed noises as follows:
\begin{equation*}
    \begin{aligned}
        I(0) &= X_1- 0.1  X_2 + \epsilon_0, \\
        I(1) &= X_1- 1.5  X_2 + \epsilon_1, \\
        I(2) &= 20 X_1- X_2 + 10 + \epsilon_2,
    \end{aligned}
    \qquad
        \begin{aligned}
        Y(0) &= 10 X_1+ X_2 + \epsilon_3, \\
        Y(1) &= 10 X_1+ X_2 + 8 + \epsilon_4, \\
        Y(2) &= 10 X_1+ X_2 + 18 + \epsilon_5, 
    \end{aligned}
\end{equation*}
where $\epsilon_j \sim N(0,2)$ for $j=0, \dots ,5$.
Treatment levels ($E=0,1,2$) are assigned with the probabilities:
\begin{equation*}
    \begin{split}
        \pi_0(X) &= \frac{1}{1 + \exp(l_1) + \exp(l_2)},  \\
        \pi_1(X) &= \frac{\exp(l_1)}{1 + \exp(l_1) + \exp(l_2)}, \\
        \pi_2(X) &= \frac{\exp(l_2)}{1 + \exp(l_1) + \exp(l_2)}, 
    \end{split}
\end{equation*}
where $l_1 = (5(X_1 + X_2) - 55)/ (10\sqrt2)$, and $l_2 = (3\sqrt 2 X_1 -43\sqrt 2 + 4\sqrt 2X_2)/ (10\sqrt 2)$.
\npdecimalsign{.} 
\nprounddigits{3}
To approximate the potential concentration indexes for this data generating process, we use a random sample of size one million and  obtain $G(0)\approx \numprint{0.12486}$, $\theta(1)\approx   \numprint{-0.1887868}$ and $\theta(2)\approx  \numprint{ 0.02209007}$. These are used as true values when computing biases and empirical coverages.  

We use sample sizes $n=1000,2000$, and 1000 replicates. The estimators studied are the plug-in, and two efficient influence function (EIF) based estimators: the one-step estimator \eqref{1s.eq} (denoted EIF\_1S) and the estimating equation estimator \eqref{eq.eq} (denoted EIF\_EQ).
For this purpose we use the EIF obtained in Theorem \ref{eif:thm} under Assumption \ref{alternative1.ass}, although robustness to the assumption is studied below.
The deployment of these estimators is described in the Supplementary Material, Section \ref{estimation.app}, including how the nuisance functions are fitted. All computations are done using the software R \citep{R}. 

Results on bias, standard errors (both Monte Carlo and estimated), and coverage (using the estimated standard errors) for $\theta(1),\theta(2)$ are displayed in Table \ref{tab:1} and \ref{tab:1b} for $n=1000,2000$ respectively. For completeness, we give also results on the estimation of $G(0)$ in Supplementary Material Section \ref{sim.app}, Tables \ref{tab:5} and \ref{tab:6}.
Estimated standard errors are obtained by computing the sample variance of the fitted efficient influence function. 
Note that this corresponds to the asymptotic semiparametric efficiency bound.
The plug-in estimator is on the other hand superefficient if the correct nuisance models are used \citep{moosavi2021} and thus using this same standard error may yield conservative empirical coverages in such cases. 
To investigate robustness to misspecification of the nuisance functions (Remark \ref{miss.rem}) we replace the covariates used to simulate the data with transformations: $\tilde X_j=\log X_j^2$, $j=1,2$. Thus, the tables include results on four situations: when all nuisance models are correctly specified (labeled "Correct models"); when only exposure probabilities $\pi_e(X)$ are misspecified ("Incorrect model for $\pi$"); when only outcome models, $M_e(X)$ are misspecified ("Incorrect model for $Y$"); when all nuisance functions mentioned in Assumption \ref{product.ass} are misspecified ("Incorrect model for all").
The results confirm the following expected results from the theory presented: a) the plug-in estimator performs well when $M_e(X)$ is correctly specified 
but can be seriously biased when the latter nuisance model is misspecified, yielding in such cases low coverages; b) the EIF\_1S estimator behaves well when only one of the models for $E$ or $Y$ given $X$ is misspecified, yielding low bias and correct coverages; c) none of the methods is unbiased and/or give expected coverage when all models are misspecified. Let us expand on an important aspect of point b) above: using $\tilde X_j=\log X_j^2$ in a model for $Y\mid \tilde X,E=e$ corresponds to a situations where Assumption \ref{alternative1.ass} does not hold because $\tilde X_j$ is not a one-to-one transformation of $X_j$. Thus, the results are in line with Remark \ref{miss.rem}, where we mentioned the robustness of the estimator EIF\_1S to Assumption \ref{alternative1.ass}. 

Note that EIF\_EQ does not behave as EIF\_1S in all our experiments. When the plug-in estimator is biased, while EIF\_1S corrects this bias, EIF\_EQ corrects it only partially (one can show analytically that the bias correction is approximately half the one of EIF\_1S). This is an interesting result per se since these estimators are in many situations equivalent, for instance, for the estimation of $B(\mp)$ \cite[e.g.,][]{Demystifying}.

\npdecimalsign{.}
\nprounddigits{3}

\begin{table}\caption{\label{tab:1}Sample size $n=1000$: Bias, standard errors (both estimated, est sd, and Monte Carlo, MC sd), and empirical coverages for 95\% confidence intervals over 1000 replicates.  Three cases: Correct parametric models for all nuisance functions, only $\pi$ misspecified, only model for $Y$ misspecified, and all nuisance functions misspecified.}
    \centering 
\begin{tabular}{l n{2}{3} n{2}{3} n{2}{3} n{3}{3} n{2}{3} n{2}{3} n{2}{3} l}
\hline\hline 
\multicolumn{1}{c}{} &\multicolumn{4}{c}{$\theta(1) $}  &\multicolumn{4}{c}{$\theta(2)$} \\
    \multicolumn{1}{c}{Estimator}                     & \multicolumn{1}{c}{bias}            & \multicolumn{1}{c}{MC sd}         & \multicolumn{1}{c}{est sd}       & \multicolumn{1}{c}{coverage} & \multicolumn{1}{c}{bias}            & \multicolumn{1}{c}{MC sd}         & \multicolumn{1}{c}{est sd}        & \multicolumn{1}{c}{coverage} \\
    \hline
\multicolumn{9}{c}{Correct models} \\

   Plug-in     & 0.000185254333940343 & 0.0175952520106573 & 0.018737305291427  &  .957       & 0.000636845699764521 & 0.0129038795351558 & 0.0139409380836445 & .962       \\
   EIF\_1S    & 0.00136822878935017  & 0.0188348362497189 & 0.018737305291427  & .947       & 0.00156449280318033  & 0.0140134881121671 & 0.0139409380836445 & .939       \\
 EIF\_EQ &  0.001431752 &0.01797308 &0.01864443      &.949 &0.001309048 &0.01305365 &0.01380422      &.961  \\
\hline
\multicolumn{9}{c}{Incorrect model for $\pi$} \\

   Plug-in     & 0.000272423766592977 & 0.0168233413018301 & 0.0181806207612763 & .965       & 0.000651496496148013 & 0.0127738981245749 & 0.0135112273321161 & .956       \\
   EIF\_1S       & 0.000670655102528672 & 0.0182548477864504 & 0.0181806207612763 & .947       & 0.000840924211898297 & 0.0138137215532146 & 0.0135112273321161 & .928       \\
EIF\_EQ   &.00002302871 &0.01737899 &0.01812859      &.958 &0.0005529155 &0.01272438 &0.01345255      &.955 \\
\hline
\multicolumn{9}{c}{Incorrect model for $Y$} \\

   Plug-in    & 0.105966167887004    & 0.0109573215296827 & 0.0221250593968208 & .0         & 0.00681888923391222  & 0.0120234105949686 & 0.0221279279134863 & 1.000      \\
   EIF\_1S & 0.00466522457655144  & 0.0211456089350859 & 0.0221250593968208 & .958       & 0.00585510493837474  & 0.0206972661230586 & 0.0221279279134863 & .966       \\
 EIF\_EQ   &0.05137717 &0.01485368 &0.02232346      &.317 &0.001325161 &0.01467257 &0.02225747      &.998 \\
\hline
\multicolumn{9}{c}{Incorrect model for all} \\

  Plug-in    & 0.111133078083875    & 0.0143257321911549 & 0.0227526256561265 & .0         & 0.0131803902983607   & 0.0167897652213987 & 0.0223324831507474 & .968       \\
  EIF\_1S   & 0.00665161434637107  & 0.023373576756203  & 0.0227526256561265 & .936       & 0.0205148796910261   & 0.0203127140814572 & 0.0223324831507474 & .895    \\  
   EIF\_EQ   &0.05278652 &0.01615106 &0.02285149      &.315 &0.003617918 &0.01468901  &0.0224546      &.994\\
  \hline
\end{tabular}
\end{table}

\begin{table}\caption{\label{tab:1b}Sample size $n=2000$: Bias, standard errors (both estimated, est sd, and Monte Carlo, MC sd), and empirical coverages for 95\% confidence intervals over 1000 replicates. Three cases: Correct parametric models for all nuisance functions, only $\pi$ misspecified, only model for $Y$ misspecified, and all nuisance functions misspecified.}
    \centering 
\begin{tabular}{l n{2}{3} n{2}{3} n{2}{3} n{3}{3} n{2}{3} n{2}{3} n{2}{3} n{3}{3}}
\hline\hline
\multicolumn{1}{c}{} &\multicolumn{4}{c}{$\theta(1)$}  &\multicolumn{4}{c}{$\theta(2)$} \\
    \multicolumn{1}{c}{Estimator}                     & \multicolumn{1}{c}{bias}            & \multicolumn{1}{c}{MC sd}         & \multicolumn{1}{c}{est sd}       & \multicolumn{1}{c}{coverage} & \multicolumn{1}{c}{bias}            & \multicolumn{1}{c}{MC sd}         & \multicolumn{1}{c}{est sd}        & \multicolumn{1}{c}{coverage} \\
    \hline

\multicolumn{9}{c}{Correct models} \\
 Plug-in     & 0.000523508991721711 & 0.0119745736448471  & 0.0131719945105421 & .968       & 0.000780782656440798 & 0.00892298460608717 & 0.00976426342035617 & .961       \\
EIF\_1S     & 0.0012388993354841   & 0.0126632270416635  & 0.0131719945105421 & .952       & 0.00122504427446368  & 0.00946214931249385 & 0.00976426342035617 & .951       \\
EIF\_EQ   &.00001247329 &0.01248869 &0.01319159      &.963 &0.0006908129 &0.008885754 &0.00981499      &.969 \\
\hline
\multicolumn{9}{c}{Incorrect model for $\pi$} \\

Plug-in      & 0.000134111608159751 & 0.0125237479403676  & 0.0128301260008077 & .956       & 0.000256890638563489 & 0.00912095504685571 & 0.00951428621775183 & .958       \\
EIF\_1S & 0.000450691734868613 & 0.0134636204029603  & 0.0128301260008077 & .940       & 0.000573162304902092 & 0.00986204647360485 & 0.00951428621775183 & .936       \\
EIF\_EQ   &0.0008755285 &0.01295122 &0.01284357      &.949 &0.0004864682 &0.009313775 &0.009546809      &.959\\
\hline
\multicolumn{9}{c}{Incorrect model for $Y$} \\

Plug-in   & 0.106406242066052    & 0.00784332649252163 & 0.015705804708534  & .0         & 0.00596645089746425  & 0.00879531596863967 & 0.015652917257772   & .998       \\
EIF\_1S & 0.00437144702300074  & 0.0148398061092649  & 0.015705804708534  & .962       & 0.00573725243116352  & 0.0142823799216086  & 0.015652917257772   & .951       \\
EIF\_EQ   &0.05204728 &0.01024644 &0.01576835       &.026 &0.001138834 &0.01014959  &0.0156909      &.994\\
\hline
\multicolumn{9}{c}{Incorrect model for all} \\

Plug-in    & 0.112992430914285    & 0.0101546124169942  & 0.0160344943732491 & .0         & 0.013365839386792    & 0.0118174407261631  & 0.0157741375167042  & .936       \\
EIF\_1S   & 0.00686342689453401  & 0.0153563620108288  & 0.0160344943732491 & .944       & 0.0207340124911383   & 0.0134100865405526  & 0.0157741375167042  & .792    \\ 
EIF\_EQ   &0.05312864 &0.01065322 &0.01603876       &.024 &0.003861742 &0.0100146 &0.01580976      &.998\\
\hline  
\end{tabular}
\end{table}

\newpage

\section{Education effect on income-related health inequality}\label{app.sec}

We study here the effect of education on income-related health inequality. We use a database linking a set of socioeconomic and health registers on the whole Swedish population \citep{simsam:2016}. We focus on the cohort of all those born in Sweden during 1950 and living in Sweden in 2000, since we use as follow-up period the years 2000-2010 (i.e., between age 50 and 60). Focusing on a single cohort mitigates issues related to interference between units when it comes for instance to potential incomes.   

As health variable $Y$ we consider several measures, although our main analysis will be based on survival, i.e. $Y=1$ for individual dying during 2000-2010 and $Y=0$ otherwise. 
We also perform analyses for the two other health variables measured during the follow-up period: the number of hospitalization days (inpatient, all diagnoses), and the number of hosptialization days (inpatient, cardio-vascular diagnoses). 
As income measure $I$, we consider the average annual income while alive during 2000-2010. To stabilize the variance of $I$ we use the transformation $(\cdot + 1)^{0.2}$.

The education exposure $E$ has 3 levels (9 years or less schooling, 2-3 years senior high-school, or higher). We measure education level in the data in the year 1990. 
Pre-exposure covariates $X$ (potential confounders) are: Sex;  Living place at age 15 (year 1965; rural municipalities $= 1$, mixed  $= 2$, large cities  $= 3$; according to the Swedish Agency for Economic and Regional Growth, \citeauthor{Tillv}, \citeyear{Tillv}); Father's income, accumulated between years 1968-1970 (using the same transformation than for $I$); Place where the mother lived in 1968 (rural municipalities $= 1$, mixed  $= 2$, large cities  $= 3$).
The size of the cohort is 104172, after removing 763 because education was missing, and 1 because some covariates were missing. Further, 410 individuals are discarded in order to enforce the positivity Assumption \ref{overlap.ass}, i.e. they had too small fitted propensity scores ($\hat\pi_e(x)<0.01$ for $e=0,1$ or 2).

\npdecimalsign{.}
\nprounddigits{3}

\begin{table}[H]
\caption{ \label{tab:Death:CI} Cohort 1950, health variable survival: Naive, plug-in and one-step EIF-based estimates of the counterfactual concentration index $G(e)$, together with their standard errors (sd); the latter obtained using the influence function for the one-step estimator (yields a conservative sd for the plug-in estimator).}\label{tab:concentration1}
 \centering
 \begin{tabular}{ll n{2}{3} n{2}{3} n{2}{3}  n{2}{3} n{2}{3} }
\hline\hline
\multicolumn{1}{c}{E} & \multicolumn{1}{c}{n}    & \multicolumn{1}{c}{Naive}  & \multicolumn{1}{c}{sd}  & \multicolumn{1}{c}{Plug-in}  & \multicolumn{1}{c}{EIF\_1S}      &   \multicolumn{1}{c}{sd}       \\
\hline
0                     & 25343                 & -0.4242022576          &0.012503008         & 0.02698795308           & -0.4126122485                               & 0.01825503695         
 \\
1                     & 47435                 & -0.3882284667         &0.011990627          & 0.03819109673           & -0.3542933152                                & 0.01501226493          
\\
2                     & 30984                 & -0.3231537962        &0.019974175           & 0.02924321266           & -0.2988022031                             & 0.02471702375         
\\
\hline
\end{tabular}  
\end{table}

\begin{table}[H]\caption{ \label{tab:Death:contrasts}Cohort 1950, health variable survival: Naive, plug-in and one-step EIF-based estimates of the causal effects of education, together with standard errors (sd).}
    \centering 
   \begin{tabular}{l n{2}{3} n{2}{3} n{2}{3} n{2}{3} n{2}{3}  }
\hline\hline
    & \multicolumn{1}{c}{Naive} & \multicolumn{1}{c}{sd} & \multicolumn{1}{c}{Plug-in} & \multicolumn{1}{c}{EIF\_1S} &    \multicolumn{1}{c}{sd} \\
\hline
$\theta(1)$  &0.035973790865596 &0.017323404  &0.01120314365                & 0.05831893335                               & 0.02370433286 
\\
$\theta(2)$   &0.101048461380499 &0.023564653 &0.002255259583               & 0.1138100455                                 & 0.0307720666 
\\
\hline

\end{tabular}
\end{table}

We present in Tables \ref{tab:Death:CI} and  \ref{tab:Death:contrasts} 
the obtained concentration indexes and the contrasts $\widehat\theta(1)$ and $\widehat\theta(2)$ for the survival health variable (results for the hospitalization variables are provided in Supplementary Material Section \ref{hospital.sec}). Implementation details are given in Supplementary Material Section \ref{estimation.app}. 
We see that the naive concentration indexes in Table \ref{tab:Death:CI} indicate decreasing inequality with increasing education.
In order to control for confounding we use both the plug-in and the one-step EIF-based estimator, the latter having the best performance in our simulation study. We observe that controlling for confounding decreases the value of the concentration index slightly for all education levels compared to the naive estimation. Moreover, as in the simulation study, the plug-in estimator shrinks all estimates towards zero. The contrasts estimated in Table \ref{tab:Death:contrasts} indicate that the decrease in inequality between the lowest education level and the middle and highest level is larger after controlling for confounding and that this decreased inequality is significant at the 5\% level.
When looking at the two alternative health variables (number of hospitalization days for all diagnoses and for CVD diagnoses only, see Tables \ref{tab:Hosp:CI}-\ref{tab:CVD:contrasts} in Supplementary Material Section \ref{hospital.sec}) we also observe health inequalities for hospitalization after correcting for confounding, although those inequalities are much smaller (yet significant at the 5\% level) for hosptialization due to CVD only.

\color{black}

\section{Conclusions}\label{conc.sec}
We have introduced the theory and method to study counterfactual concentration indexes. This is useful in observational studies of socioeconomic-related health inequalities, where one want to study the effect of exposures and interventions.  
In particular, we propose estimators which are regular asymptotic linear and locally efficient under certain conditions. We give their asymptotic normal distribution including the asymptotic variance which can be estimated thereby allowing for valid inference. The estimators proposed require the fit of several nuisance functions, and they have some rate robustness properties, where some of these functions can be fitted at a lower convergence rate than the parametric $\sqrt n$-rate. We run finite sample experiments illustrating the behavior of the estimators and their inference. 

An interesting observation is the difference in behaviour between the one-step estimator and the estimating equation estimator. While both are based on the efficient influence function, the latter does not fully correct for the bias of the plug-in estimator when nuisance models are misspecified. As we have noted, verifying the assumptions of Theorem \ref{remainder.thm} is not trivial for the estimating equation estimator, and we suspect that the observed difference in bias is due to the assumptions of Theorem \ref{remainder.thm} not holding in the misspecified situations studied in our numerical experiments. 

In a case study we have applied the methods proposed to study income-related health inequalities for different levels of education. The results indicate that such inequalities exist after correcting for confounding in the Swedish cohort born in 1950, and that inequality decreases for increasing education level.
An important caveat, as for all observational studies, is that there may still be omitted confounders, which bias these results. A future direction for research is to develop a valid sensitivity analysis to the ignorability Assumptions \ref{unconfoundedness.ass} and \ref{alternative2.ass}). For this purpose, one could try and generalize recent recent results by \citet{scharfstein2021semiparametric} or \cite{moosavi:2024} obtained in a semiparametric context for the classical average causal effect estimand.

	

\newpage

\bibliographystyle{agsm}
\bibliography{bibliography.bib}
\setcounter{section}{0}

\def\spacingset#1{\renewcommand{\baselinestretch}%
{#1}\small\normalsize} \spacingset{1}

\date{}
\if0\blind
{
  \title{\bf Causal inference targeting a concentration index for studies of health inequalities}
  \author{Mohammad Ghasempour\thanks{
    We are grateful for comments given by Tetiana Gorbach which have improved the paper. The Marianne and Marcus Wallenberg Foundation is acknowledged for their financial support.}\hspace{.2cm}\\
    Department of Statistics/USBE, Ume\r{a} University, Sweden \\
    and \\
    Xavier de Luna \\
    Department of Statistics/USBE, Ume\r{a} University, Sweden 
    \\
    and \\
    Per E. Gustafsson \\
    Department of Epidemiology and Global Health, Ume\r{a} University, Sweden}
  \maketitle
} \fi

\if1\blind
{
  \bigskip
  \bigskip
  \bigskip
  \begin{center}
    {\LARGE\bf Supplementary Material for manuscript: Causal inference targeting a concentration index for studies of health inequalities}
\end{center}
  \medskip
} \fi



\spacingset{1.9} 
\npdecimalsign{.} 
\nprounddigits{3}

\section{Proof of Theorem \ref{eif:thm}: Efficient influence functions}\label{proofthm1.app}
Consider the path $\mp_{s}=s\mathbbm{1}_{\tilde o}+(1-s)\mp$, where $s \in [0,1]$, and $\mathbbm{1}_{\tilde o}$ is the point mass distribution  \citep[e.g.,][]{Demystifying}. Then, for $j=1,2$, the efficient influence function for $\Psi_j$ can be derived as follows:
\begin{align*}
    \varphi_j(\tilde o,\Psi_j)=\left.\frac{d \Psi_j \left(\mp_{s}\right)}{d s}\right|_{s=0}=2\frac{\frac{d A_j}{d s}|_{s=0} \cdot {B}|_{s=0} - \frac{d B}{d s}|_{s=0} \cdot A_j|_{s=0}}{B^{2}|_{s=0}}\\
    :=2\frac{\ph{{A_j}}{\tilde{o}}{\mp}}{B(\mp)}-2\frac{\ph{B}{\tilde{o}}{\mp} \cdot A_j(\mp)}{B(\mp)^2}.
\end{align*}
Let us first deduce the influence function under Assumptions \ref{unconfoundedness.ass}, \ref{alternative2.ass}, \ref{overlap.ass}. 
By taking the Gâteaux derivative of our estimator at a distribution in the defined path, with respect to the path parameter $s$, we have:

\begin{align*}
        \ph{{A_2}}{\tilde{o}}{\mp}=&\left.\frac{d A_2}{d s}\right|_{s=0}=\left.\frac{d}{d s}\left(\int y \Xi_{e}(i,\mp_s) \frac{f_{s}(y, i, x, E=e)}{f_{s}(x,E=e )}  {f_{s}(x)} d y d i d x\right) \right|_{s=0}\\
        =&\left.\int y \Xi_{e}(i,\mp_s) \left[ \frac{\frac{df_s (y,i,x,E=e)}{ds}  f_s(x) + f_s(y,i, x,E=e) \frac{df_s(x)}{ds}}{ f_s(x,E=e) } \right. \right. \\
        &-\left. \left.\frac{f_{s}(y, i, x,E=e) f_{s}(x)\left[\frac{d f_{s}( x,E=e)}{d s}\right]}{f_{s}(x,E=e)^{2} } \right] dy di dx \right|_{s=0}\\
        &+\left.\int y \frac{\frac{d \Xi_{e}(i,\mp_s)}{d s} f_{s}(y, i,x,E=e) f_{s}(x)}{f_{s}(x,E=e) } dydidx\right|_{s=0}.
\end{align*}  
Evaluating the above derivative at the point $s=0$, yields to:
\begin{align*}        
        &\int y \frac{\Xi_{e}(i,\mp) f(y, i, x,E=e) f(x)}{f( x,E=e) }\\
        &\cdot\left[\frac{\mathbbm{1}(y, i, x,E=e)}{f(y, i,  x,E=e)}-1+\frac{\mathbbm{1}(x)}{f(x)}-1 - \frac{\mathbbm{1}(E=e, x)}{f( x,E=e)}+1 \right] d y d i d x\\
         &+ \EE_\mp ( \EE_\mp \{ Y  \frac{d \Xi_{e}(I,\mp_s)}{d s} \mid  X,E=e\}) \mid_{s=0}.
\end{align*}        
The second term above is the sensitivity of $\Xi_{e}(i,\mp_s)$. Recall,
\begin{align*}
    \Xi_{e}(i,\mp_s)=&\EE_{\mp_s}\{\EE_{\mp_s}(\mathbbm{1}(I\leq i) \mid X, E=e) \}.
    \end{align*}
    Then,
\begin{align*}
    \left.\frac{d  \Xi_{e}(j,\mp_s)}{d s} \right|_{s=0}  =&  \int \mathbbm{1}(i\leq j) \frac{ f(i, E=e, x) f(x)}{f(E=e, x) }\left[\frac{\mathbbm{1}(i, E=e, x)}{f(i, E=e, x)}-1+\frac{\mathbbm{1}(x)}{f(x)}-1 - \frac{\mathbbm{1}(E=e, x)}{f(E=e ,x)}+1 \right] d i d x \\
    =& \frac{\mathbbm{1}(E=e)}{\pi_e(\tilde{x})} [ \mathbbm{1}(\tilde{i} \leq j) - \EE_\mp(\mathbbm{1}(I \leq j) \mid  \tilde{x}, E=e) ] + \EE_\mp(\mathbbm{1}(I \leq j) \mid \tilde{x},E=e) - \Xi_{e}(j,\mp)\\
    =:&\varphi_\Xi (j,\tilde o,\mp).
\end{align*}
Substituting the latter into $\left.\frac{d A_2(\mp_s)}{ds}\right|_{s=0}$:

\begin{align*}
\ph{{A_2}}{\tilde{o}}{\mp} = &
\frac{\mathbbm{1}(\tilde e=e)}{\pi_e( \tilde{x})} [\tilde{y} \Xi_{e}(\tilde{i},\mp)  -  \EE_\mp(Y  \Xi_{e}(I,\mp) \mid  \tilde{x},E=e)]  +\EE_\mp(Y \ \Xi_{e}(I,\mp) \mid  \tilde{x},E=e) -  A_2(\mp) \\
 &+ \EE [ \EE \{  Y \varphi_\Xi (I,\tilde o,\mp) \mid   X,E=e \} ] .
\end{align*}

Term $B$ is typical when the average causal effect is the parameter of interest and its derivative has been deduced elsewhere \citep[e.g.,][]{Demystifying}:
\begin{align*}
    \varphi_B(\tilde{o},{\mp})=\left.\frac{d B}{d s}\right|_{s=0}=\frac{\mathbbm{1}(T = e) }{\pi_e( \tilde{x})}\left(\tilde{y}-\EE_\mp\left(Y \mid \tilde{x}, E=e \right)\right)+ \EE_\mp \left(Y |  \tilde{x},E=e\right) - B(\mp),
\end{align*}
thereby the results for $\varphi_2(\tilde o,\Psi_2)$.

Finally, under Assumption \ref{alternative1.ass}, i.e. assuming $Y(e)$ and $I(e)$ independent given $X$ lead to a separation of the inner conditional expectations:
\begin{align*}
\mathbb{E}_\mp(Y  \mathbbm{1}{\{I \leq j\}} \mid X, E=e) =
\mathbb{E}_\mp(Y \mid  X,E=e)  \mathbb{E}_\mp(\mathbbm{1}{\{I \leq j\}} |  X,E=e).
\end{align*}
Similar calculations to the above ones yield $\ph{1}{\tilde{o}}{\mp}$.

\section{Proof of Theorem \ref{remainder.thm}: Controlling the remainder term}\label{proofthm2.app}

    We can write:

\begin{align*}
\sqrt{n} & \left( \Psi_j(\pn ) + \E_\mp \left[ \varphi_j (O, \hat{\mp}_n) \right] - \Psi_j( \mp ) \right) =- \sqrt{n} \e{\mp}{\ph{j}{O}{\pn} + \ps{j}{\pn} - \ps{j}{\mp}} \\
    =&- \sqrt{n} \e{\mp}{ \frac{B(\pn) \ph{A}{O}{\pn}  - A_j(\pn) \ph{B}{O}{\pn} }{ B(\pn)^2 } + \frac{A_j(\pn)}{B(\pn)}- \frac{A_j(\mp)}{B(\mp)} }  \\
    =&- \sqrt{n} \left\{ \e{\mp}{ \frac{ \ph{A}{O}{\pn} + A_j(\pn) - A_j(\mp) }{ B(\pn) } }    - \e{\mp}{ \frac{A_j(\pn) ( \ph{B}{O}{\pn} + B(\pn) - B(\mp)) }{ B(\pn)^2 } }  \right. \\ 
    & + \left.  \frac{A_j(\pn)}{B(\pn)} -  \frac{A_j(\mp)}{B(\mp)}  + \frac{A_j(\mp) - A_j(\pn)}{B(\pn)} +  \frac{A_j(\pn) (B(\pn) - B(\mp)) }{B(\pn)^2}  \right\} \\
    =&- \sqrt{n} \left[ \frac{1}{B(\pn)} R_A(\mp,\pn) - \frac{A_j(\pn)}{B(\pn)^2} R_B(\mp,\pn) + \frac{B(\mp) -B(\pn)}{B(\pn)} \left( 
 \frac{A_j(\mp)}{B(\mp)} -  \frac{A_j(\pn)}{B(\pn)} \right) \right] \\
=&- \sqrt{n} \left[ \frac{1}{B(\pn)} R_A(\mp,\pn) - \frac{A_j(\pn)}{B(\pn)^2} R_B(\mp,\pn) \right]+  o_\mp(1),
\end{align*}
where we have let 
$R_A(\mp,\pn)=A_j(\pn)-A_j(\mp)+\EE_\mp(\ph{{A_j}}{O}{\pn})$ and $R_B(\mp,\pn)=B_j(\pn)-B_j(\mp)+\EE_\mp(\ph{{B_j}}{O}{\pn})$. The last equality holds by Assumption \ref{product.ass}i).


Now because of Assumption \ref{product.ass}ii), it remains to show that $R_A(\mp,\pn)$ and $R_B(\mp,\pn)$ are $o_\mp(\sqrt{n}^{-1})$ to complete the proof. Note that these are the two remainder term of a von Mises expansion \eqref{vonM.eq} for the parameters $A_j(\mp)$ and $B(\mp)$ respectively.

For $B(\mp)$, this term has been widely studied, and we know that $R_B(\mp,\pn)=o_\mp(\sqrt{n}^{-1})$ under Assumption \ref{product.ass}iii); for more details we refer to Section 4.4 of  \citet{Demystifying}.

We know consider first the case $A_1(\mp)$. We can write:

\begin{align*}
R_A&(\mp,\pn)= \e{\mp}{\ph{{A_1}}{O}{\pn} + A_1(\pn) - A_1(\mp)} \\
=\ & \EE_\mp\left[ \frac{\ie}{\hat\pi_e( X)} \left[Y \Xi_{e}(I,\pn)  -  \hat M_e(X)\hat\EE_{\pn (\tilde I\mid X)}( \Xi_{e}(\tilde I,\pn) \mid  X,E=e)\right]  \right]\\
+\ & \EE_{\mp}\left[\hat M_e(X)\hat\EE_{\pn(\tilde I\mid X)}( \Xi_{e}(\tilde I,\pn) \mid  X,E=e)-A_1(\pn)  \right ] \\
+\ & \EE_\mp\left[\hat\EE_{\pn (\tilde X)} \left[\hat M_e(\tilde X) \hat\EE_{\pn(\tilde I\mid \tilde X)}\left[ \varphi_\Xi(\tilde I, O,\pn)\mid \tilde X, E=e  \right] \right] \right] \\
  +\ & \e{\mp}{ A_1(\pn)- A_1(\mp)  }.\\
\end{align*}

Note that the tilde notation above is necessary to introduce in order to keep track for what random variables the expectations are taken over, especially for the third term after the equality sign, noting that $\EE_\mp$ is taken only over $O$.

After some algebraic manipulation we get
\begin{align*}
R_A&(\mp,\pn) \\
=\ & \e{\mp}{ \hat\EE_{\pn(\tilde X)} \left[ \hat{M}_e(\tilde X) \left( \frac{\pi_e(X)}{\hat{\pi}_e(X)}  -1\right) \hat{\mathbb{E}}_{\pn(\tilde I \mid \tilde X)}\left[ \left. N_e(\tilde I, X)- \hat{N}_e(\tilde  I, X) \right| \tilde X, E=e \right]  \right]  }\\
+\ & \e{\mp}{ \left( \frac{\pi_e(X)}{\hat{\pi}_e(X)} -1 \right) \left( M_e(X)\mathbb{E}_{\mp }(\Xi_e(I,\pn)|X,E=e)- \hat{M}_e(X) \mathbb{E}_{\hat{\mathcal P}_n (\tilde I \mid X)}(\Xi_e( \tilde I,\pn)|X,E=e) \right)   }\\
+\ & \e{\mp}{ M_e(X)  \e{\mp}{ \Xi_e(I,\pn) - \Xi_e(I,\mp) | X , E=e }}\\
+\ & \e{\mp}{ \hat\EE_{\pn(\tilde X)}{\left[ \hat{M}_e(\tilde X) \hat\EE_{\pn( \tilde I \mid \tilde X)}{\left[ N_e(\tilde I,X) - \Xi_e(\tilde I,\pn) |\tilde X, E=e \right]} \right]} }
\end{align*}
Note that to get the second term we use Assumption \ref{product.ass}viii). It is an immediate result of Assumption \ref{product.ass}iv) and v) that the first two terms are $o_\mp(n^{-1/2})$. The last two terms in the equation above can be re-written as:
\begin{align*}
 &  \e{\mp}{ M_e(X)  \e{\mp}{ \Xi_e(I,\pn) | X , E=e }  }  - \e{\mp}{ M_e(X)  \e{\mp}{ \Xi_e(I,\mp) | X , E=e }  } \\
 +\ & \EE_{\mp}{\left[  \hat{M}_e(\tilde X) \hat\EE_{\pn( \tilde I \mid \tilde X)}{\left[ \Xi_e(\tilde I,\mp)  |\tilde X, E=e \right]}  \right]} - \EE_{\mp}{\left[  \hat{M}_e(\tilde X) \hat\EE_{\pn( \tilde I \mid \tilde X)}{\left[ \Xi_e(\tilde I,\pn)  |\tilde X, E=e \right]}  \right]}\\
 +\ & \EE_{\mp}{\left[  \hat{M}_e(\tilde X) \hat\EE_{\pn( \tilde I \mid \tilde X)}{\left[ \Xi_e(\tilde I,\pn)  |\tilde X, E=e \right]}  \right]} - \EE_{\mp}{\left[  \hat{M}_e(\tilde X) \hat\EE_{\pn( \tilde I \mid \tilde X)}{\left[ \Xi_e(\tilde I,\mp)  |\tilde X, E=e \right]}  \right]}\\
+\ & \hat\EE_{\pn(\tilde X)}{\left[  \hat{M}_e(\tilde X) \hat\EE_{\pn( \tilde I \mid \tilde X)}{\left[ \Xi_e(\tilde I,\mp)  |\tilde X, E=e \right]}  \right]} - \hat\EE_{\pn(\tilde X)}{\left[  \hat{M}_e(\tilde X) \hat\EE_{\pn( \tilde I \mid \tilde X)}{\left[ \Xi_e(\tilde I,\pn)  |\tilde X, E=e \right]}  \right]} 
\end{align*}
The last two lines are an empirical process that vanishes asymptotically under Assumption \ref{product.ass}x) \citep[][Lemma 19.24]{Vaart_1998}.
To deal with the the first two lines, we rewrite them by first defining a linear functional:
$$\Lambda(\Xi_e(\cdot,\pn) - \Xi_e(\cdot,\mp) ),$$
where
$$\Lambda(h ) : = \e{\mp}{ M_e(X)  \e{\mp}{ h(I) | X , E=e }  -  \hat{M}_e( X) \hat\EE_{\pn( \tilde I \mid  X)}{\left[ h(\tilde I)  | X, E=e \right] }   } .$$
Using the Riesz representation theorem, there exists a function $\lambda(O)$ such that:
$$ \Lambda(h ) = \langle \lambda \,,h\rangle  = \EE{[\lambda(O)h(O)]}, \forall h,$$
and 
$$||\Lambda||_{sup} = ||\lambda|| .$$
Therefore, we have 
$$\Lambda(\Xi_e(\cdot,\pn) - \Xi_e(\cdot,\mp) ) = \EE{[\lambda(O)[\Xi_e(I,\pn) - \Xi_e(I,\mp)]]} \leq \EE{[(\Xi_e(I,\pn) - \Xi_e(I,\mp))^2]^\frac{1}{2}} ||\Lambda||_{sup}^\frac{1}{2},$$
where in the last inequality we have used Cauchy-Schwarz inequality. The right-hand-side of the inequality is $o_\mp(n^{-1/2})$ by Assumption ix).

We now consider the remainder term for $A_2(\mp)$:

\begin{align*}
R_A&(\mp,\pn)= \e{\mp}{\ph{{A_2}}{O}{\pn} + A_2(\pn) - A_2(\mp)} \\
=\ & \EE_\mp\left[ \frac{\ie}{\hat\pi_e( X)} \left[Y \Xi_{e}(I,\pn)  -  \hat\EE_{\pn(\tilde Y, \tilde I\mid X)}(\tilde Y \Xi_{e}(\tilde I,\pn) \mid  X,E=e)\right]  \right]\\
+\ & \EE_{\mp}\left[\hat\EE_{\pn(\tilde Y, \tilde I\mid X)}(\tilde Y  \Xi_{e}(\tilde I,\pn) \mid  X,E=e)-A_2(\pn)  \right ] \\
+\ & \EE_\mp\left[\hat\EE_{\pn(\tilde X)} \left[ \hat\EE_{\pn(\tilde Y, \tilde I\mid \tilde X)}\left[\tilde Y \varphi_\Xi(\tilde I, O,\pn)\mid \tilde X, E=e  \right] \right] \right] \\
  +\ & \e{\mp}{ A_2(\pn)- A_2(\mp)  }.\\
\end{align*}
By considering a similar procedure to the one for $A_1(\mp)$ we have:
\begin{align*}
R_A&(\mp,\pn) \\
=\ & \e{\mp}{ \hat\EE_{\pn(\tilde X)}{\left[  \left( \frac{\pi_e(X)}{\hat{\pi}_e(X)}  -1\right) \hat{\mathbb{E}}_{\pn(\tilde Y ,\tilde I \mid \tilde X)}\left[ \tilde Y (\left. N_e(\tilde I, X)- \hat{N}_e(\tilde  I, X) )\right| \tilde X, E=e \right]  \right]}  }\\
+\ & \e{\mp}{ \left( \frac{\pi_e(X)}{\hat{\pi}_e(X)} -1 \right) \left( \mathbb{E}_{\mp }(Y \Xi_e(I,\pn)|X,E=e)-  \mathbb{E}_{\hat{\mathcal P}_n (\tilde Y, \tilde I \mid X)}(\tilde Y \Xi_e( \tilde I,\pn)|X,E=e) \right)   }\\
+\ &  \e{\mp}{ \e{\mp}{Y \Xi_e(I,\pn) | X , E=e }  }  - \e{\mp}{   \e{\mp}{ Y \Xi_e(I,\mp) | X , E=e }  } \\
 +\ & \EE_{\mp}{\left[  \hat\EE_{\pn(\tilde Y, \tilde I \mid \tilde X)}{\left[ \tilde Y \Xi_e(\tilde I,\mp)  |\tilde X, E=e \right]}  \right]} - \EE_{\mp}{\left[  \hat\EE_{\pn(\tilde Y,  \tilde I \mid \tilde X)}{\left[\tilde Y \Xi_e(\tilde I,\pn)  |\tilde X, E=e \right]}  \right]}\\
 +\ & \left\{ \EE_{\mp}-\hat\EE_{\pn(\tilde X)}\right\}{\left[ \hat\EE_{\pn( \tilde Y, \tilde I \mid \tilde X)}{\left[\tilde Y \Xi_e(\tilde I,\pn)  |\tilde X, E=e \right]} -  \hat\EE_{\pn(\tilde Y, \tilde I \mid \tilde X)}{\left[ \tilde Y \Xi_e(\tilde I,\mp)  |\tilde X, E=e \right]}  \right]}. 
\end{align*}

As earlier, the last term in the equation above is an empirical process that vanishes asymptotically under Assumption \ref{product.ass}xii).
It is an immediate result of Assumption \ref{product.ass}vi) and vii) that the first two terms are $o_\mp(n^{-1/2})$.
Using Assumption \ref{product.ass}xi) and the Riesz representation theorem in the same way as in the earlier case ensures that the second and third terms in the equation above is also $o_\mp(n^{-1/2})$.

\section{Estimation}\label{estimation.app}

We give here details on three estimation strategies introduced in the paper, for the parameters $\Psi_1$ and $\Psi_2$, i.e. the plug-in estimator, the one-step estimator \eqref{1s.eq} and the estimating equation estimator \eqref{eq.eq}.

\subsection{Plug-in estimator}
The plug-in estimator is $\Psi_j(\pn)=2\frac{A_j(\pn)}{B(\pn)}-1$. Term $B$ is a usual potential outcome mean. Its plug-in estimation is obtained by a model driven or data-adaptive fit $\hat M_e(X)$, and averaging over the predicted outcomes for the whole sample:
\begin{align*}
B(\pn)= \mathbb{E}_n[\hat M_e(X)],
\end{align*}
where $\EE_n V=\frac{1}{n} \sum_{i=1}^nV_i$.

The terms $A_j$s are not trivial to estimate. 
We give here details for $A_1$.
\begin{align*}
A_1(\pn) = \EE_n[\hat M_e(X)  \hat\EE(\Xi_{e}(I,\pn) \mid  X,E=e)].
\end{align*}
We need a fit of $\EE(\Xi_{e}(I,\pn) \mid X, E=e)$. 
One possibility is here to rewrite
\begin{align}\label{eq:lastmodel}
 \hat\EE_{\pn}[\Xi_{e}(I,\pn) \mid X,E=e] =&\hat\EE_{\pn(I\mid X)}\left[ \hat\EE_{\pn(\tilde X)} \left( \hat\EE_{\pn(\tilde I\mid \tilde X)} (\mathbbm{1}(\tilde I\leq I)\mid \tilde X, \tilde E=e) 
 \right) \Bigg | X, E=e  \right] \nonumber \\
    =&\hat\EE_{\pn(\tilde X)} \left[ \hat\EE_{\pn(I\mid X)}\left \{\hat\EE_{\pn(\tilde I\mid \tilde X)} (\mathbbm{1}(\tilde I\leq I)\mid \tilde X, \tilde E=e)\mid X, E=e \right \}\right].
\end{align}
We propose to fit the last term with the following algorithm:
\begin{enumerate}
\item Comparing all pairs of observed incomes in subsample with $E=e$ yields a $n_e\times n_e$ matrix  (where $n_e$ individuals have treatment level $e$). For all entry in the matrix (all pairs of income $(i,\tilde i)$) regress $\mathbbm{1}(\tilde i\leq i)$ on $x_i,x_{\tilde i}$, using probit link. Note this is a correctly specified model for the DGP of the simulation study.

\item Use the model fitted in 1. to predict $\mathbbm{1}(\tilde i\leq i)$ for all pairs of individuals in the sample, yielding a $n \times n$ matrix of fitted values.

\item Take the average of the row in this matrix for which $X=x_i$, to obain an implementation of the estimator (\ref{eq:lastmodel}); $\EE_{\pn}[\Xi_{e}(I,\pn) \mid X =x_i,E=e]$.

\end{enumerate}

\subsection{One-step estimator}

The one-step estimator \eqref{1s.eq} of $G(e)$ is the sum of the plug-in estimator and the average of estimated influence function over the data:
\begin{equation*}
\begin{array}{l}
\Psi_j(\pn) + E_n[\ph{j}{O}{\pn}] = 2\frac{A(\pn)}{B(\pn)}\left( \frac{E_n[\ph{A}{O}{\pn}]}{A(\pn)}-\frac{E_n[\ph{B}{O}{\pn}]}{B(\pn)} +1 \right)-1 \\
= 2\frac{A(\pn)}{B(\pn)}\left( \frac{E_n[\ph{A}{O}{\pn}] + 2A(\pn)}{A(\pn)}-\frac{E_n[\ph{B}{O}{\pn}] + B(\pn)}{B(\pn)}  \right)-1
\end{array}
\end{equation*}

Under Assumption \ref{alternative1.ass} (i.e., for $\Psi_1$), five nuisance functions need to be fitted to obtain $\ph{A}{O}{\pn}$ (see Theorem \ref{eif:thm}). Those include the potential health outcome model $M_e(X)$, the score $\pi_e(X)$, the counterfactual ranks of incomes
$\Xi_{e}(I,\mp)$, as well as:
$$\EE_\mp( \Xi_{e}(I,\mp) \mid  X,E=e) $$
and 
$$ 
\EE_{\mp(X)} [ \EE_{\mp(I \mid X)} \{ \EE_{\mp(\tilde I \mid \tilde X)}  ( \mathbbm{1}(\tilde I\leq I)\mid \tilde X, \tilde E=e ) \mid  X, E=e  \} ].
$$

The last two items can be fitted with the same method as in the previous section, i.e. by regressing $\mathbbm{1}(i\leq \tilde i)$ on $x_i,x_{\tilde i}$, while $\Xi_{e}(i,\pn)$ can be fitted with the following algorithm for a given income $i$:

\begin{enumerate}
    \item Start by regressing $\mathbbm{1}( \tilde I\leq i)$ on $\tilde{X}$
for $(\tilde{X},\tilde{I}) \in \{(x_j,i_j), j=1,\dots,n_e\}$ (subsample with E=e), with, e.g., a logit link.
    \item Use the fitted model to predict the values: $\EE(\mathbbm{1}( \tilde I\leq i) \mid \tilde{X})$ for the whole sample $\{(x_j,i_j), j=1,\dots,n\}$.
    \item Take the average $\EE_n\{\hat\EE_{\pn}(\mathbbm{1}( \tilde I\leq i) \mid \tilde{X})\}$, to obtain the estimation of $\Xi_{e}(i,\pn)$. 
\end{enumerate}
The terms $A(\pn)$ and $B(\pn)$ are the plug-in estimators described in the previous section. Finally, $\ph{B}{O}{\pn}$ is estimated using $\hat{\pi}_e( X)$ and $\hat{M}_e(X)$.

\subsection{Estimating equation estimator}

The estimating equation estimator solves \eqref{eq.eq}. We can write:
\begin{equation*}
\EE_n[\ph{j}{O}{\pn}] = 0 \Rightarrow \frac{A(\pn)}{B(\pn)}\left( \frac{E_n[\ph{A}{O}{\pn}]}{A(\pn)}-\frac{E_n[\ph{B}{O}{\pn}]}{B(\pn)}\right) = 0.
\end{equation*}
Thus for the uncentered influence function, the equation is:
\begin{equation*}
\begin{array}{l}
\EE_n[\ph{j}{O}{\pn}] + 2\frac{A(\pn) }{B(\pn)} = 2\frac{A(\pn) }{B(\pn)} \\
\Rightarrow \frac{A(\pn)}{B(\pn)}\left( \frac{E_n[\ph{A}{O}{\pn}]}{A(\pn)}-\frac{E_n[\ph{B}{O}{\pn}]}{B(\pn)}\right) + 2\frac{A(\pn) }{B(\pn)}= 2\frac{A(\pn) }{B(\pn)}\\
\Rightarrow \frac{A(\pn)}{B(\pn)} \frac{E_n[\ph{A}{O}{\pn}]}{A(\pn)}  + 2\frac{A(\pn) }{B(\pn)}= \frac{A(\pn) }{B(\pn)} \left(\frac{E_n[\ph{B}{O}{\pn}]}{B(\pn)} + 2 \right) .
\end{array}
\end{equation*}
The equation can be solved by equating the following ratios:
 \begin{align*}
          &\frac{\EE_n[\ph{A}{O}{\pn}] + 2A(\pn) }{B(\pn)} = \frac{A(\pn) }{B(\pn)}\left[\frac{\EE_n[\ph{B}{O}{\pn}] + 2B(\pn) }{B(\pn)} \right] \\
         \Rightarrow &
          \widehat{\left (\frac{ A}{ B}\right)} = \frac{\EE_n[\ph{A}{O}{\pn}] + 2A(\pn)}{\EE_n[\ph{B}{O}{\pn}] + 2 B(\pn) }.
\end{align*}
The latter yields an estimator for $\Psi$ when using the fits of nuisance models described in the previous section.

\subsection{Further details on fitting the nuisance functions}
In the simulation study, a multinomial logistic regression is used to fit the propensity score, and a linear model is used for the potential health outcome. The remaining nuisance models are fitted as described above.

In the case study, the only differences are that for the potential health outcomes, a zero-inflated Poisson regression including all second order interaction terms is considered with $L_1$ regularization, using the R-package \texttt{mpath}, and similarly for the multinomial model for the propensity score, using the R-package \texttt{glmnet}. All analyses in this case study are done using R \citep{R}, and SparkR \citep{spark}. The latter is needed to fit the binary response regressions in the first algorithm described above. 
For instance, for the second education level the size of the training set is 2'250'079'225, which makes it almost impossible to do it on typically available RAM when using regular packages in R.

\section{Further simulation results}\label{sim.app}

The following tables present the simulation results on the estimation of $G(0)$ for samples of size 1000, and 2000, as a complement of the results presented in the manuscript.

\begin{table}[H]\caption{\label{tab:5}Sample size $n=1000$: Bias, standard errors (both estimated, est sd, and Monte Carlo, MC sd), and empirical coverages for 95\% confidence intervals for estimation of $G(0)$ over 1000 replicates. Three cases: Correct parametric models for all nuisance functions, only $\pi$ misspecified, only model for $Y$ misspecified, and all nuisance functions misspecified.}
\centering 
\begin{tabular}{l n{2}{3} n{2}{3} n{2}{3} n{2}{3}}\hline \hline
 \multicolumn{1}{c}{Estimator} & \multicolumn{1}{c}{bias}  & \multicolumn{1}{c}{MC sd} & \multicolumn{1}{c}{est sd} & \multicolumn{1}{c}{coverage} \\
 \hline
\multicolumn{5}{c}{Correct models} \\
 Plug-in   & 0.000942 & 0.01379  & 0.014627 & .956 \\
 EIF\_EQ    & 0.001636 & 0.013936 & 0.014627 & .959 \\
 EIF\_1S    & 0.002456 & 0.014564 & 0.014627 & .949 \\
 \hline
\multicolumn{5}{c}{Incorrect model for $\pi$} \\
 Plug-in   & 0.000359 & 0.01327  & 0.014283 & .959 \\
 EIF\_EQ    & 0.000777 & 0.01349  & 0.014283 & .953 \\
 EIF\_1S    & 0.00132  & 0.014154 & 0.014283 & .937 \\
 \hline
\multicolumn{5}{c}{Incorrect model for $Y$} \\
 Plug-in   & 0.070483 & 0.008236 & 0.018989 & .001 \\
 EIF\_EQ    & 0.033421 & 0.011555 & 0.018989 & .591 \\
 EIF\_1S    & 0.006323 & 0.017731 & 0.018989 & .955 \\
 \hline
\multicolumn{5}{c}{Incorrect model for all} \\
 Plug-in   & 0.079468 & 0.011377 & 0.019481 & 0    \\
 EIF\_EQ    & 0.031449 & 0.01261  & 0.019481 & .690 \\
 EIF\_1S    & 0.016885 & 0.018539 & 0.019481 & .887\\
 \hline
\end{tabular}
\end{table}

\begin{table}[H]\caption{\label{tab:6}Sample size $n=2000$: Bias, standard errors (both estimated, est sd, and Monte Carlo, MC sd), and empirical coverages for 95\% confidence intervals for estimation of $G(0)$ over 1000 replicates. Three cases: Correct parametric models for all nuisance functions, only $\pi$ misspecified, only model for $Y$ misspecified, and all nuisance functions misspecified.}
\centering 
\begin{tabular}{l n{2}{3} n{2}{3} n{2}{3} n{2}{3}}\hline \hline
 \multicolumn{1}{c}{Estimator} & \multicolumn{1}{c}{bias}  & \multicolumn{1}{c}{MC sd} & \multicolumn{1}{c}{est sd} & \multicolumn{1}{c}{coverage} \\
 \hline
\multicolumn{5}{c}{Correct models} \\
Plug-in & 0.0000392 & 0.009529 & 0.010382 & .964  \\
EIF\_EQ  & 0.000525 & 0.009603 & 0.010382 & .969  \\
EIF\_1S  & 0.001072 & 0.009986 & 0.010382 & .960  \\
 \hline
\multicolumn{5}{c}{Incorrect model for $\pi$} \\
Plug-in & 0.000677 & 0.00981  & 0.010125 & .950  \\
EIF\_EQ  & 0.000473 & 0.01004  & 0.010125 & .948  \\
EIF\_1S  & 0.000208 & 0.010564 & 0.010125 & .939  \\
 \hline
\multicolumn{5}{c}{Incorrect model for $Y$} \\
Plug-in & 0.071084 & 0.005831 & 0.013413 & 0    \\
EIF\_EQ  & 0.033713 & 0.008176 & 0.013413 & .205  \\
EIF\_1S  & 0.006319 & 0.012443 & 0.013413 & .942  \\
 \hline
\multicolumn{5}{c}{Incorrect model for all} \\
Plug-in & 0.079651 & 0.008106 & 0.013741 & 0    \\
EIF\_EQ  & 0.031283 & 0.008647 & 0.013741 & .325  \\
EIF\_1S  & 0.017364 & 0.012532 & 0.013741 & .785 \\
 \hline
\end{tabular}
\end{table}

\section{Further empirical results}\label{hospital.sec}
The following tables present the results on the estimation of the counterfactual concentration index and the resulting contrasts when: the health outcome is the number of hospitalization days (in-patient, all diagnoses, Tables \ref{tab:Hosp:CI} and \ref{tab:Hosp:contrasts}) during 2000-2010, and the number of hosptialization days (in-patient, cardio-vascular diagnoses only, Tables \ref{tab:CVD:CI} and \ref{tab:CVD:contrasts}) during the same period. 

\begin{table}[H]
\caption{ \label{tab:Hosp:CI}Cohort 1950, health variable is number of hospitalization days: Naive, plug-in and the two proposed EIF-based estimators of the counterfactual concentration index, together with the standard deviation obtained using the influence function (valid for the EIF-estimators and conservative for Plug-in).}\label{tab:concentration2}
 \centering
 \begin{tabular}{ll n{2}{3} n{2}{3} n{2}{3}  n{2}{3}  n{2}{3}}
\hline\hline
\multicolumn{1}{c}{E} & \multicolumn{1}{c}{n}    & \multicolumn{1}{c}{Naive}  &   \multicolumn{1}{c}{sd}  & \multicolumn{1}{c}{Plug-in}  & \multicolumn{1}{c}{EIF\_1S}    &   \multicolumn{1}{c}{sd}        \\
\hline
0 & 25343 & -0.3768   & 0.010350233  &-0.000741062 & -0.3619     & 0.019232  \\
1 & 47435 & -0.3718   & 0.011988473  &0.011461047  & -0.37342    & 0.020105  \\
2 & 30984 & -0.34043  & 0.014969135  &0.002866941  & -0.34625    & 0.022702 
\\
\hline
\end{tabular}  
\end{table}

\begin{table}[H]\caption{ \label{tab:Hosp:contrasts}Cohort 1950, health variable is number of hospitalization days: Plug-in and two proposed EIF-based estimators of the causal effects of education, together with the standard deviation obtained using the influence function (valid for the EIF-estimators and conservative for Plug-in).}
    \centering 
   \begin{tabular}{l  n{2}{3} n{2}{3} n{2}{3} n{2}{3}  n{2}{3}}
\hline\hline
    & \multicolumn{1}{c}{Naive} & \multicolumn{1}{c}{sd} & \multicolumn{1}{c}{Plug-in} & \multicolumn{1}{c}{EIF\_1S}  & \multicolumn{1}{c}{sd}  \\
\hline
$\theta(1)$   &0.005009 & 0.015838271 & 0.012202 & -0.01152  & 0.027871 \\
$\theta(2)$   &0.036378 & 0.018198965 & 0.003608 & 0.01565   & 0.029794 
\\
\hline

\end{tabular}
\end{table}

\begin{table}[H]
\caption{ \label{tab:CVD:CI}Cohort 1950, health variable is number of CVD hospitalization days: Naive, plug-in and the two proposed EIF-based estimators of the counterfactual concentration index, together with the standard deviation obtained using the influence function (valid for the EIF-estimators and conservative for Plug-in).}\label{tab:concentration3}
 \centering
 \begin{tabular}{ll n{2}{3} n{2}{3} n{2}{3}  n{2}{3} n{2}{3}}
\hline\hline
\multicolumn{1}{c}{E} & \multicolumn{1}{c}{n}    & \multicolumn{1}{c}{Naive} & \multicolumn{1}{c}{sd}    & \multicolumn{1}{c}{Plug-in}  & \multicolumn{1}{c}{EIF\_1S}     &   \multicolumn{1}{c}{sd}      \\
\hline
0 & 25343 & -0.082577966798854 &0.015193628 & 0.0459678338170437 & -0.0632944161447413  & 0.017059567
 \\
1 & 47435 & -0.0649622124252104 &0.012176786 & 0.058091934824941  & -0.0673226812263473  & 0.012306266
 \\
2 & 30984 & -0.0366915094189988 &0.018260751 & 0.0677379361551107 & -0.0494280522580732  & 0.018742987
\\
\hline
\end{tabular}  
\end{table}

\begin{table}[H]\caption{ \label{tab:CVD:contrasts}Cohort 1950, health variable is number of CVD hospitalization days: Plug-in and two proposed EIF-based estimators of the causal effects of education, together with the standard deviation obtained using the influence function (valid for the EIF-estimators and conservative for Plug-in).}
    \centering 
   \begin{tabular}{l n{2}{3} n{2}{3}  n{2}{3} n{2}{3} n{2}{3}}
\hline\hline
    & \multicolumn{1}{c}{Naive} & \multicolumn{1}{c}{sd} & \multicolumn{1}{c}{Plug-in} & \multicolumn{1}{c}{EIF\_1S}  & \multicolumn{1}{c}{sd} \\
\hline
$\theta(1)$   &0.0176157543736436 &0.019471016 &0.012124101
 	&-0.00402826508160603 	   &0.021040246
\\
$\theta(2)$   &0.0458864573798552 &0.023755028 &0.021770102338067 	&0.0138663638866681 	&0.02534782
\\
\hline

\end{tabular}
\end{table}



\end{document}